%% file: main.tex
\documentclass[11pt,letter]{article}

\usepackage{amsmath,amssymb,amsthm}
\usepackage{thmtools}
\usepackage{thm-restate}
\usepackage[margin=1in]{geometry}
\usepackage[utf8]{inputenc}
\usepackage{booktabs}
\usepackage{amsmath}
\usepackage[colorlinks=true,linkcolor=red,citecolor=blue]{hyperref}
\usepackage{amssymb}
\usepackage{tikz}
\usetikzlibrary{topaths,calc}
\usepackage{thm-restate,thmtools}
\usepackage{mathtools}
\usepackage{amsthm}
\usepackage{dsfont}
\usepackage{color}
\usepackage{amsmath}
\usepackage{amsfonts}
\usepackage{amssymb}
\usepackage{bbm}
\usepackage{relsize}
\usepackage{enumitem}
\usepackage[procnumbered, linesnumbered,algoruled]{algorithm2e}
\usepackage{lineno}
\usepackage{array}

\usepackage[nameinlink]{cleveref}

\newcommand{\comment}[1]{}

\newcommand{\mm}{M_{n,k,\alpha}(\Pi)}
\newcommand{\opt}{\textnormal{opt}}
\newcommand{\bases}{\textnormal{bases}}
\newcommand{\IS}{\textnormal{IS}}
\newcommand{\cA}{{\mathcal{A}}}
\newcommand{\cB}{{\mathcal{B}}}
\newcommand{\cQ}{{\mathcal{Q}}}
\newcommand{\cG}{{\mathcal{G}}}
\newcommand{\cI}{{\mathcal{I}}}

\newcommand{\cD}{{\mathcal{D}}}
\newcommand{\cm}{{\mathcal{M}}}

\newcommand{\cF}{{\mathcal{F}}}
\newcommand{\cJ}{{\mathcal{J}}}

\newcommand{\fs}{f_{\textnormal{SAT}}}
\newcommand{\io}{|I|^{O(1)}}

\newcommand{\cP}{{\mathcal{P}}}
\newcommand{\cC}{{\mathcal{C}}}

\newcommand{\OPT}{\textnormal{OPT}}

\newcommand{\eps}{{\varepsilon}}

\newcommand{\problem}[2]{
	\noindent
	\begin{tabular}{ p{0.8in} p{5.4in} }
		\\
		\hline
		\multicolumn{2}{c}{#1} \\
		\hline
		#2\\
		\hline
		\\
	\end{tabular}
}

\newtheorem{lemma}{Lemma}[section]

\newtheorem{definition}[lemma]{Definition}
\newtheorem{theorem}[lemma]{Theorem}

\newtheorem{observation}[lemma]{Observation}
\newtheorem{claim}[lemma]{Claim}
\crefname{claim}{claim}{claims}

\def \II   {{\mathcal I}}

	\crefname{claim}{claim}{claims}
		\crefname{cor}{corollary}{corollaries}
\begin{document}

	
		\title{
			  Lower Bounds for Matroid Optimization Problems with \\ a Linear Constraint
		}
		
		\author{Ilan Doron-Arad\thanks{Computer Science Department, 
				Technion, Haifa, Israel. \texttt{idoron-arad@cs.technion.ac.il}}
			\and
			Ariel Kulik\thanks{CISPA Helmholtz Center for Information Security, Saarland Informatics Campus, Germany. \texttt{ariel.kulik@cispa.de}} 
			\and
			Hadas Shachnai\thanks{Computer Science Department, 
				Technion, Haifa, Israel. \texttt{hadas@cs.technion.ac.il}}
		}
		\maketitle
		\thispagestyle{empty}
		
		\begin{abstract}
			\input{abstract}

		\end{abstract}
		
		\newpage

		\setcounter{page}{1}
		
		\input{intro}

		\input{results}
		
				\input{technical}

		\input{Pi-matroid}

\input{HardnessOracle}

\input{SAT-matroid}

\input{discussion}

		\bibliographystyle{splncs04}
		\bibliography{bibfile}
		\appendix
		
	\end{document}

%% file: abstract.tex
We study a family of {\em matroid optimization problems with a linear constraint} (MOL). 
In these problems, we seek a subset of elements which optimizes (i.e., maximizes or minimizes) a linear objective function subject to (i) a matroid independent set, or a matroid basis constraint, (ii) additional linear constraint. A notable 
member in this family is {\sc budgeted matroid independent set (BM)}, which  can be viewed as classic $0/1${\sc-knapsack} with a matroid constraint. While special cases of BM, such as  {\sc knapsack with cardinality constraint} and {\sc multiple-choice knapsack}, admit a {\em fully polynomial-time approximation scheme} (Fully PTAS), the best known result for BM on a general matroid is an {\em Efficient} PTAS. Prior to this work, the existence of a Fully PTAS for BM, and more generally, for any problem in the  family of MOL problems, 
has been open.

In this paper, we answer this question negatively by showing that none of the (non-trivial) problems in this family admits a Fully PTAS. This resolves the complexity status of several well studied problems. Our main result is obtained by showing first that {\sc exact weight matroid basis (EMB)} does not admit a pseudo-polynomial time algorithm. This distinguishes EMB from the special cases of $k$-{\sc subset sum} and EMB on a linear matroid, which are solvable in pseudo-polynomial time.
We then obtain unconditional hardness results for the family of MOL 
problems in the oracle model (even if randomization is allowed), and show that the same results hold when the matroids are encoded as part of the input, assuming $\textnormal{P} \neq \textnormal{NP}$. For the hardness proof of EMB, we introduce the $\Pi$-{\em matroid} family. This intricate subclass of matroids, which exploits the interaction 
between a weight function and the matroid constraint, may find use in tackling other matroid optimization problems.

%% file: intro.tex
	\section{Introduction}
	\label{sec:intro}
	
	Matroids are simple combinatorial structures, providing
	a unified abstraction for independence systems such as linear independence in a vector space, or cycle-free subsets of edges in a given graph. A {\em matroid} is a set system $(E, \cI)$, where $E$ is a finite set and $\cI \subseteq 2^E$ are the {\em independent sets (IS)} such that $(i)$ $\emptyset \in \cI$, $(ii)$ for all $A \in \cI$ and $B \subseteq A$ it holds that $B \in \cI$, and $(iii)$ for all $A,B \in \cI$ where $|A| > |B|$, there is $e \in A \setminus B$ such that $B \cup \{e\} \in \cI$.\footnote{Properties $(ii)$ and $(iii)$ are known as {\em hereditary property},
		and  {\em exchange property}.}
	
	While serving as a generic abstraction for numerous applications, matroids possess useful combinatorial properties that allow the development of efficient algorithms. 
	These algorithms include such canonical results as the classic greedy approach for finding a maximum weight independent set of a matroid (see, e.g., \cite{CLRS22}), Edmond's algorithm for matroid partitioning \cite{edmonds1965minimum}, and Lawler's algorithm for matroid intersection \cite{La75}. In all of the above, polynomial running time is enabled due to the structure of the problem $-$ a single objective function with a matroid constraint. However, in many natural applications, there is an added {\em linear} constraint. 
	
	Consider, for example, the problem of finding a maximum independent set in a matroid subject to a budget constraint. Formally, we are given 
	a set of elements $E$, a membership oracle for a collection of independent sets $\II \subseteq 2^{E}$ of a matroid $(E, \II)$, 
	a budget $L>0$, a weight function~$w:E\to \mathbb{R}_{\geq 0}$, and a value function $v:E\rightarrow \mathbb{R}_{\geq 0}$. A {\em solution} for the problem is an independent set $S \in \II$ of total weight at most $L$, i.e., $w(S) \leq L$.\footnote{For every set $X$, a function $f:X \rightarrow \mathbb{R}_{\geq 0}$ and $Y \subseteq X$ we define $f(Y) = \sum_{e \in Y} f(e)$.} The {\em value} of a solution $S$ is given by $v(S)$,
	and the objective is to find a solution of maximum value. This problem, known as {\sc budgeted matroid independent set (BM)}, is a generalization of the classic $0/1${\sc-knapsack}, which is NP-hard
	and therefore unlikely to admit an exact polynomial-time algorithm. Thus, obtaining efficient approximations has been a main focus in the study of BM.

 For an instance $I$ of an optimization problem $\cG$, let $\OPT_{\cG}(I)$ be the value of an optimal solution for $I$. For some $\rho \geq 1$, a {\em $\rho$-approximate solution} $S$ for $I$ is a solution of value $v \geq \frac{\OPT_{\cG}(I)}{\rho}$ ($v \leq \rho \cdot \OPT_{\cG}(I)$) if $\cG$ is a maximization (minimization) problem. 
 We say that $\cA$ is a {\em randomized $\rho$-approximation algorithm} for $\cG$ if,
given an instance $I$ of $\cG$, $\cA$
returns with probability at least $\frac{1}{2}$ a $\rho$-approximate solution for $I$ $-$ if a solution exists. If no solution exists $-$ $\cA$ returns that $I$ does not have a solution.

	Let $|I|$ be the encoding size of an instance~$I$ of a problem $\cG$. 
	A (randomized) {\em polynomial-time approximation scheme} (PTAS)
	for $\cG$ is a family of algorithms $(A_{\eps})_{\eps>0}$ such that, for any $\eps>0$, $A_{\eps}$ is a (randomized) polynomial-time $(1+ \eps)$-approximation algorithm for $\cG$. 
	A (randomized) {\em Efficient PTAS} (EPTAS) is a (randomized) PTAS $(A_{\eps})_{\eps>0}$ with running time of the form $f\left(\frac{1}{\eps}\right) \cdot |I|^{O(1)}$, where $f$ is an arbitrary computable function. 
	The strong dependence of run-times on the error parameter, $\eps > 0$, often renders the above schemes highly impractical. This led to the study of the following more desirable class of schemes. A (randomized) approximation scheme 
	$(A_{\eps})_{\eps>0}$ is a (randomized) {\em Fully PTAS} (FPTAS) if the running time of $A_{\eps}$ is of the form $ {\left(\frac{|I|}{\eps}\right)}^{O(1)}$.\footnote{To better distinguish between EPTAS and FPTAS, we use throughout the paper Efficient PTAS and Fully PTAS.}
	
In the past decades, BM was shown to admit a PTAS~\cite{BBGS11,CVZ11,GZ10}, and more recently an Efficient PTAS \cite{DKS23,DKS23b}. As the special case of
$0/1${\sc-knapsack} admits a Fully PTAS, it is natural to explore the existence of a Fully PTAS for BM. There are known Fully PTASs  for BM on restricted families of matroids. This includes {\sc knapsack with a cardinality constraint}~\cite{caprara2000approximation}, 
{\sc multiple-choice knapsack}~\cite{La79}, and {\sc BM with laminar matroid constraint} \cite{DKS23c}. However, the question of whether BM admits a Fully PTAS on general matroids remained open. 

In this paper, we resolve this question negatively for BM and other fundamental matroid optimization problems with a linear constraint. 

\subsection{Our Results}

For a matroid $\cm = (E,\cI)$ we define $\IS(\cm) = \cI$ and $\bases(\cm) = \{S \in \cI~|~|S| = \textsf{rank}(\cm)\}$, where $\textsf{rank}(\cm) = \max_{T \in \cI} |T|$ is the {\em rank} of $\cm$, i.e., the maximum cardinality of an independent set. We study a family of {\em matroid optimization problems with a linear constraint} (MOL). Problems in this family are characterized by three parameters: 
\begin{enumerate}
\item[(i)]
The optimization objective $\opt$ $-$ either the operator $``\max"$ or $``\min"$. 
\item[(ii)]
A {\em matroid feasibility constraint} $\cF$ $-$ either the independent sets of a matroid, or the set of {\em bases} of a matroid. The feasibility constraint is $\cF \in \{\IS,\bases\}$.
\item[(iii)]
A relation $\triangleleft$ $-$  realized by one of the relations $``\geq"$ or $``\leq"$. 
\end{enumerate}

Let $\cP = \{\max,\min\} \times \{\bases,\IS\} \times \{\leq ,\geq\}$ be the set of {\em parameters} for MOL problems. 
Based on the set of parameters $\cP$, we define for every triplet a problem
in the MOL family. For $P \in \cP$ where $P = \left(\opt,\cF,\triangleleft\right)$, define the $P$-{\sc matroid optimization with a linear constraint} ($P$-MOL) problem as follows. An instance is a tuple $I = (E,\cI,v,w,L)$ such that $\cm = (E,\cI)$ is a matroid, $v:E \rightarrow \mathbb{R}_{\geq 0}$ is the objective function, $w:E \rightarrow \mathbb{R}_{\geq 0}$ is a weight function, and $L \in \mathbb{R}_{\geq 0}$ is a {\em bound} for the linear constraint. 
A {\em solution} of $I$ is $S \subseteq E$ which satisfies the matroid feasibility constraint $S \in \cF(\cm)$ and the linear constraint $w(S) \triangleleft L$. The goal is to optimize (i.e., maximize or minimize) the value $v(S)$. Thus, we can formulate a P-MOL optimization problem as
\begin{equation}
\label{eq:0}
\opt~ v(S) \text{ s.t. } S \in \cF(\cm),~  w(S) \triangleleft L.
\end{equation}

Observe that $(\max,\IS,\leq)$-MOL is the BM problem.
That is, given a BM instance (equivalently, $(\max,\IS,\leq)$-MOL instance) $I = (E,\cI,v,w,L)$, the goal is to find an independent set $S \in \cI$ of maximum total value $v(S)$ such that $w(S) \leq L$. Other notable examples for MOL problems are {\sc constrained minimum basis of a matroid} (CMB) \cite{HL04}, which can be cast as $(\min,\bases,\leq)$-MOL, and {\sc knapsack cover with matroid constraint} (KCM) \cite{CCNR13} formalized by $(\min,\IS,\geq)$-MOL.\footnote{CMB, KCM, and other MOL problems may not have a solution; however, we can decide in polynomial time if a solution exists, and our definition of approximation algorithms captures instances with no solution.}

We note that \eqref{eq:0} does not refer to the {\em representation} of the instance $I$. We consider two possible representations.
For any $P \in \cP$, in an instance $(E,\cI,v,w,L)$ of {\em oracle} $P$-MOL, the arguments
$E,v,w,L$ are given as the input, and the independent sets $\cI$ are accessed via a {\em membership oracle}, which determines whether a given set $S \subseteq E$ belongs to  $\cI$ in a single query. Thus, the independent sets are not considered in the encoding size of the instance. The term {\em running time} for problems involving oracles refers to the sum of the number of queries to the oracle and the number of basic operations. 
Previous works on MOL problems often consider membership oracles \cite{CVZ11,BBGS11,CCNR13,DKS23,DKS23b}. As hardness with oracles does not necessarily imply hardness in non-oracle models (see, e.g., \cite{canetti2004random,chang1994random}), in \Cref{sec:SAT} we show lower bounds for variants of MOL problems in which the independent sets are encoded as part of the input.

Clearly, the problem $(\min,\IS,\leq)$-MOL is trivial since the empty set achieves the optimal objective value. 
However, for any other $P \in \cP$, solving the $P$-MOL problem is challenging.
The {\em non-trivial} MOL problems are all the MOL problems excluding $(\min,\IS,\leq)$-MOL. That is, $P$-MOL is non-trivial if $P\in \cQ$ where  $\cQ = \cP \setminus \{(\min,\IS,\leq)\}$.
 Observe that non-trivial MOL problems are NP-hard (e.g., $0/1${\sc-knapsack} is a special case of $(\max,\IS,\leq)$-MOL);
however, 
all previously studied MOL problems admit approximation schemes.

For certain special cases of MOL problems, e.g., BM with simple matroid constraints,
the existence of a Fully PTAS is known for decades \cite{Si79,caprara2000approximation}.
However, for MOL problems with arbitrary matroid constraints, the best known results are Efficient PTAS.  While matroids form an important generalization of well known basic
constraints,
the complexity of the corresponding MOL problems remained open. Specifically, prior to this work, the existence of a MOL problem which does not admit a Fully PTAS was open.

Our main result is that none of the (non-trivial) 
oracle matroid optimization problem with a linear constraint admits 
a Fully PTAS, even if randomization is allowed. This unconditioned hardness result is established by deriving a lower bound on the minimum number of queries to the membership oracle. 
 \begin{restatable}{theorem}{thmMain}
	\label{thm:main}
	For every $P \in \cQ$ there is no randomized  \textnormal{Fully PTAS} for oracle $P$-\textnormal{MOL}. 
\end{restatable} 

\begin{table}[tbh]
	\centering
	\begin{tabular}{m{2.7in} m{1.4in} m{1.9in}}
		\toprule
		Problem & Previous Results & This Paper \\
		\midrule
		Budgeted Matroid Independent Set & Efficient PTAS \cite{DKS23} & {\bf No Fully PTAS} \\
		Budgeted Matroid Intersection & Efficient PTAS  \cite{DKS23b} & {\bf No Fully PTAS}  \\
		Constrained Minimum Basis of Matroid & Efficient PTAS  \cite{HL04} & {\bf No Fully PTAS} \\
		Knapsack Cover with a Matroid & PTAS \cite{CCNR13} & {\bf No Fully PTAS}  \\
		\bottomrule
	\end{tabular}
	\caption{Implications of our results for previously studied MOL problems. All of our bounds hold for randomized algorithms. }
	\label{tab:results}
\end{table}

\Cref{thm:main} conclusively distinguishes MOL problems with arbitrary matroids, such as BM, from special cases with simpler matroid constraints. Furthermore, it shows that existing Efficient PTAS \cite{HL04,DKS23,DKS23b} for MOL problems on general matroids are the best possible. Notable implications of our results are given in \Cref{tab:results}, and consequences of our lower bounds for a set of previously studied problems~\cite{ CGM92,HL04,BBGS11,CVZ11,GZ10,CCNR13,DKS23,DKS23b,doron2023budgeted}
are given in \Cref{sec:results}. By resolving the complexity status of MOL problems on general matroids, our results promote future research to design (or rule out) Fully PTAS for MOL problems on restricted matroid classes (see \Cref{sec:discussion}). 

To prove \Cref{thm:main}, we turn our attention to the following problem. 
\begin{definition}
	\label{def:EMB}
	An instance of {\sc Exact Matroid Basis (EMB)}  is $I =  (E,\cI,c,T)$, where $(E,\cI)$ is a matroid, $c:E \rightarrow \mathbb{N}$ is a weight function, and $T \in \mathbb{N}$ is a target value. A {\em solution} is a basis $S$ of $(E,\cI)$ such that $c(S)=T$. The goal is to decide if there is a solution. 
\end{definition}
Similar to MOL problems, EMB does not specify the input. 
In an instance $(E,\cI,c,T)$ of {\sc oracle-EMB},
$E,c,T$ are explicitly given, and the independent sets $\cI$ are accessed via a membership oracle. An instance $I$ of a decision problem $\cD$
is a ``yes"-instance if the correct answer for $I$ is ``yes"; otherwise, $I$ is a ``no"-instance.
We say that $\cA$ is a {\em randomized algorithm} for a decision problem $\cD$ if,
 given a ``yes"-instance $I$ of $\cD$, $\cA$ returns ``yes"  with probability at least $\frac{1}{2}$; for a ``no"-instance, $\cA$ returns ``no" with probability $1$. The next result rules out a pseudo-polynomial time algorithm for oracle-EMB, thus distinguishing the problem from the special cases of $k$-{\sc subset sum} and EMB on linear matroids, which admit a pseudo-polynomial  time algorithm \cite{CGM92}. 

\begin{restatable}{theorem}{thmemb}
	\label{thm:emb}
	There is no randomized  algorithm  for {\sc oracle-EMB} that for any  \textnormal{oracle-EMB} instance $I =(E,\cI,c,T)$, where $n = |E|+1$ and  $m = c(E)+1$,  runs in time $ \left(n \cdot(T+2) \cdot m \right)^{O(1)}$. 
\end{restatable}

\subsection{Technical Contribution}
We derive our results by introducing $\Pi$-{\em matroids}.
This new family of matroids carefully exploits a simple weight function to define a matroid that successfully {\em hides} a specific property $\Pi$ within its independent sets (see \Cref{sec:technical}). Using $\Pi$-matroids, we define 
	oracle-EMB instances whose solutions must satisfy the property $\Pi$. This shows the unconditional hardness of oracle-EMB, as $\Pi$ can be discovered only via an exponential number of queries to the membership oracle. Our hardness results for MOL problems (as stated in \Cref{thm:main}) are derived via reduction from oracle-EMB.

		Despite the abundance of lower bounds for matroid problems~
		\cite{lovasz1978matroid,JK82,karp1985complexity,soto2014simple}, as well as for knapsack  problems~\cite{chekuri2005polynomial,kulik2010there,cygan2019problems,bringmann2021fine}, we are not aware of lower bounds that leverage the interaction between the matroid constraint and the additional linear constraint required for deriving our new lower bound for EMB, and consequently for MOL problems. Indeed, if the matroid constraint is removed from \eqref{eq:0} (equivalently, $\cF(\cm)=2^E$), MOL problems become variants of classic 
	$0/1${\sc-knapsack}, which admits a Fully PTAS. Alternatively, if the linear constraint imposed by $w,L$ is removed, then we have the polynomially solvable maximum/minimum weight matroid independent set problem. This distinguishes our construction from existing lower bounds for matroid problems. 
	 $\Pi$-matroids may be useful for deriving lower bounds for other problems 
	(see \Cref{sec:discussion}).

		 Our unconditional lower bounds apply in the {\em oracle} model, where the independent sets of the given matroid can be accessed only via a membership oracle. One may question the validity of the bounds for variants of the problems where the matroid is encoded as part of the input. Indeed, in some scenarios, the use of oracles makes problems harder \cite{canetti2004random,chang1994random}. Thus, we complement our results by showing that the same lower bounds hold under the standard complexity assumption $\textnormal{P}\neq\textnormal{NP}$, even if the matroid is encoded as part of the instance and membership can be decided in polynomial time. We accomplish this by designing the family of {\em \textnormal{SAT}-matroids} $-$ a counterpart of the $\Pi$-matroid family whose members can be efficiently encoded. This construction can be used to obtain hardness results for other matroid problems in non-oracle models, based on existing analogous lower bounds in the oracle model (e.g., \cite{JK82}). We elaborate on that in \Cref{sec:SAT}.

%% file: results.tex
\subsection{Implications of Our Results and Prior Work}
\label{sec:results}		

Below we describe in further detail the implications of our results, and discuss previous work on MOL problems. In the following problems, general matroids are assumed to be accessed via membership oracles. 

\paragraph{Exact Matroid Basis (EMB):}
This is a generalization of the $k$-{\sc subset sum} problem (where $(E,\cI)$ is a uniform matroid).\footnote{In a uniform matroid, $\cI= \{S \subseteq E~|~ |S| \leq k\}$.} Thus, EMB is unlikely to be solvable in polynomial time. Instead, we seek a pseudo-polynomial time algorithm whose running time has polynomial dependence on the encoding size of the instance and the target value $T$. 
Indeed, the special case of EMB in which the matroid is {\em representable} (or, linear) admits such  a pseudo-polynomial time algorithm \cite{CGM92}. Since the 1990s, it has been an open question whether the result of Camerini et al.~\cite{CGM92} can be extended to general matroids. 
\Cref{thm:emb} resolves this question, ruling out the existence of  a pseudo-polynomial time algorithm for EMB.

\paragraph{Budgeted Matroid Independent Set (BM):} 

This problem is cast as $(\max,\IS,\leq)$-MOL. BM is a natural generalization of the classic $0/1${\sc-knapsack} problem, for which a Fully PTAS has been known since the 1970s~\cite{La79}. As mentioned above, a Fully PTAS is known also for other special cases of BM. A PTAS for BM was first given in \cite{BBGS11} as a special case of {\sc budgeted matroid intersection (BMI)}. In this generalization of BM, we are given {\em two} matroids $(E,\cI_1)$ and $(E,\cI_2)$, and a solution has to be an independent set of {\em both} matroids. A PTAS for BM also follows from 
the results of~\cite{GZ10,CVZ11} which present PTASs for multi-budgeted variants of BM. An Efficient PTAS for BM was recently given in \cite{DKS23} and for BMI in \cite{DKS23b}. The existence of a Fully PTAS for BM was posed as a central open question in~\cite{BBGS11,DKS23,DKS23b}. We answer this question negatively, as formalized in \Cref{thm:main}, giving a tight lower bound for BM and BMI. 

\paragraph{Constrained Minimum Basis of a Matroid (CMB):} This problem can be cast as the matroid optimization problem $(\min,\bases,\leq)$-MOL. 
The {\sc constrained minimum spanning tree (CST)} problem is the  special case of CMB in which the matroid $(E,\cI)$ is graphical \cite{RG96, andersen1996bicriterion,HL04,hong2004fully}, namely, there is a graph $G=(V,E)$ such that  
the independent sets $\cI$ are cycle-free subsets of edges in $G$.
A PTAS for CST was given by Ravi and Goemans~\cite{RG96}. This result was improved to an Efficient PTAS by Hassin and Levin \cite{HL04}. A {\em bicriteria} FPTAS, which violates the budget constraint by a factor of $(1+\eps)$, was presented in~\cite{hong2004fully}. The authors of~\cite{HL04} mention that their result actually gives an Efficient PTAS for CMB. 
The existence of a Fully PTAS for CMB remained an open question. \Cref{thm:main} shows that the Efficient PTAS for CMB cannot be improved.
 
\paragraph{Knapsack Cover with a Matroid (KCM):}
As a final implication, \Cref{thm:main} rules out the existence of a Fully PTAS for a  {\em coverage} variant of $0/1${\sc-knapsack}, formulated as $(\min,\IS,\geq)$-MOL. 
In \cite{CCNR13}, Chakaravarthy et al. presented a PTAS for KCM using integrality properties of a linear programming formulation of KCM. Moreover, for the special case of KCM with a {\em partition matroid}, they give a Fully PTAS based on dynamic programming. The existence of  a Fully PTAS for KCM  on a general matroid was posed in~\cite{CCNR13} as an open question. \Cref{thm:main} answers this question negatively.
Our initial study indicates that an Efficient PTAS  for KCM can potentially be obtained by adapting the approach of Hassin and Levin \cite{HL04} to the setting of KCM. This suggests that our lower bound cannot be strengthened.

\paragraph{Organization:}
\label{sec:organization}

In \Cref{sec:technical} we introduce the $\Pi$-matroid family and give the proof of \Cref{thm:emb}. In \Cref{sec:tight_absolute} we prove \Cref{thm:main}, and in
 \Cref{sec:SAT} we show that similar lower bounds hold in the standard computational model. We conclude in \Cref{sec:discussion} with a summary and directions for future work.

%% file: technical.tex
\section{The Hardness of {\sc oracle exact matroid basis}}
\label{sec:technical}
\label{sec:Pi}
In this section, we prove \Cref{thm:emb}. We use in the proof the family of $\Pi$-{\em matroids}. A member in the $\Pi$-matroid family is given by four arguments: $n,k,\alpha \in \mathbb{N}_{>0}$, and $\Pi \subseteq 2^{[n]}$. For any~$m\in \mathbb{N}$, let $[m]=\{1,\ldots, m\}$.
 The first argument, $n \in \mathbb{N}_{>0}$, is the number of elements, and the ground set is $[n]$. The second argument, $k \in [n]$, is the rank of the matroid. 
 The third argument, $\alpha\in \mathbb{N}_{>0}$, is a target value, that is usually equal to $\textnormal{\textsf{sum}}(S)$ for some $S \subseteq [n]$, where $\textnormal{\textsf{sum}}(S) = \sum_{i \in S} i$. The last argument is a family of subsets $\Pi \subseteq 2^{[n]}$, where $2^{[n]} = 2^{\{1,\ldots,n\}}$.

The set $\Pi$ is called the {\em secret family} because finding $S\in \Pi$ is possible only via repeated queries to the membership oracle of the matroid. Since $\Pi$ can have an arbitrary structure, this may require exhaustive enumeration.

\begin{definition}
	\label{def:Matroid}
	Let $n,k,\alpha \in \mathbb{N}_{>0}$. For some $\Pi \subseteq 2^{[n]}$, define the {\bf $\Pi$-matroid} on $n,k$, and $\alpha$ as $M_{n,k,\alpha}(\Pi) = ([n],\cI_{n,k,\alpha}(\Pi))$, where $$\cI_{n,k,\alpha}(\Pi) = \mathcal{J}_{n,k} \cup \mathcal{K}_{n,k,\alpha} \cup \mathcal{L}_{n,k,\alpha}(\Pi)$$ and $\mathcal{J}_{n,k}, \mathcal{K}_{n,k,\alpha}, \mathcal{L}_{n,k,\alpha}(\Pi)$ are defined as follows. \begin{equation}
		\label{eq:1}
		\begin{aligned}
			\mathcal{J}_{n,k}~~~~~~ = {} & \big\{S \subseteq [n]~\big|~ |S| < k\big\} \\
			\mathcal{K}_{n,k,\alpha}~~~~ = {} & \big\{S \subseteq [n]~\big|~ |S| = k, \textnormal{\textsf{sum}}(S) \neq \alpha\big\} \\
			\mathcal{L}_{n,k,\alpha}(\Pi) = {} & \big\{S \subseteq [n]~\big|~ |S| = k, \textnormal{\textsf{sum}}(S) = \alpha, S \in \Pi\big\}. \\
		\end{aligned}
	\end{equation}
\end{definition} 

\begin{figure}
	\centering
	\vspace{-4cm} 
	\begin{tikzpicture}[scale=1.4, every node/.style={draw, circle, inner sep=1pt}]
		\node (11) at (5.5,4) {$\bf \textcolor{teal}{1}$};
		\node (22) at (7,4) {$\bf \textcolor{black}{2}$};
		\node (33) at (5.5,3) {$\bf \textcolor{red}{3}$};
		\node (44) at (7,3) {$\bf \textcolor{blue}{4}$};
		\draw (11) -- (22);
		\draw (11) -- (33);
		\draw (11) -- (44);
		\draw (22) -- (44);
		\draw (33) -- (44);
		
		\node[draw=none] at (9, 3.5) {$~\Pi = \{\textnormal{independent sets of } G\}$};
		
		\node[draw=none] at (4.5, 3.5) {A graph $G$};
		\node[draw=none] at (6.75, 2.5) {$\mathcal{J}_{4,2} = \{\emptyset, \{$\bf \textcolor{teal}{1}$\},\{$\bf \textcolor{black}{2}$\},\{$\bf \textcolor{red}{3}$\},\{$\bf \textcolor{blue}{4}$\}\}, ~~\mathcal{K}_{4,2,5} = \{\{$\bf \textcolor{black}{2}$,$\bf \textcolor{blue}{4}$\},\{$\bf \textcolor{black}{2}$,$\bf \textcolor{teal}{1}$\},\{$\bf \textcolor{blue}{4}$,$\bf \textcolor{red}{3}$\},\{$\bf \textcolor{red}{3}$,$\bf \textcolor{teal}{1}$\}\}, ~~\mathcal{L}_{4,2,5}(\Pi) = \{\{$\bf \textcolor{black}{2}$,$\bf \textcolor{red}{3}$\}\}$};
	\end{tikzpicture}
	\vspace{-7.5cm} 
	\caption{\label{fig:Pi} The independent sets of the $\Pi$-matroid $M_{n,k,\alpha}(\Pi)$, with parameters $n = 4$, $k = 2$, and $\alpha = 5$. The secret family $\Pi$ contains all independent sets in the graph $G$, where $\{2,3\}$ is the only independent set in $G$ with $k$ elements.}
	
\end{figure}

In words, $\mathcal{J}_{n,k}$ contains all subsets of strictly less than $k$ elements;  $\mathcal{K}_{n,k,\alpha}$ contains all subsets of cardinality~$k$ whose total sum is not~$\alpha$; finally, 	$\mathcal{L}_{n,k,\alpha}(\Pi)$ contains all subsets of cardinality~$k$ and total sum~$\alpha$ which also belong to $\Pi$. See \Cref{fig:Pi} for an example of a member of the $\Pi$-matroid family. Using a simple argument, we show that the set system in 
\Cref{def:Matroid} is indeed a matroid.
For the sets 
$\cI_{n,k,\alpha}(\Pi)$, $ \mathcal{J}_{n,k}$, $\mathcal{K}_{n,k,\alpha}$ and $\mathcal{L}_{n,k,\alpha}(\Pi)$ defined in \Cref{def:Matroid},
we often omit the subscripts ${n,k,\alpha}$ and ${n,k}$ when the values of $n,k,\alpha$ are known by context. For simplicity, for any set $A$ and an element $a$, let $A+a$, $A-a$ be $A \cup \{a\}$ and $A \setminus \{a\}$, respectively.

\begin{restatable}{lemma}{isamatroid}
	\label{lem:1}
	For every  $n,k,\alpha \in \mathbb{N}_{>0}$, and $\Pi \subseteq 2^{[n]}$ it holds that $\mm$ is a matroid.
\end{restatable}
	\begin{proof}
	We first note that $\emptyset \in \mathcal{J}$ since $0<k$; therefore, $\emptyset \in \cI(\Pi)$. For the hereditary property, let $A \in \cI(\Pi)$. For all $B \subset A$ it holds that $|B| < k$; thus, $B \in \mathcal{J}$ and it follows that $B \in \cI(\Pi)$. For the exchange property, let $A,B \in \cI(\Pi)$ such that $|A| > |B|$. We consider the following cases. 
	\begin{enumerate}
		\item $|B| < k-1$. Then, for all $e \in A \setminus B$ it holds that $|B+e| < k-1+1 = k$; thus, $B+e \in \cJ$ and it follows that $B+e \in \cI(\Pi)$. Note that there is such $e  \in A \setminus B$ because $|A| > |B|$. 
		
		\item $|B| = k-1$ and $|A| = k$. We consider two subcases.  
		\begin{enumerate}
			\item $B \subseteq A$. Then, as $|A| > |B|$ there is $e \in A \setminus B$; thus, $B+e = A$ (because $|B| = k-1$ and $|A| = k$). As $A \in \cI(\Pi)$, it follows that $B+e \in \cI(\Pi)$.
			
			\item $B \not\subseteq A$.  Then, as $|B| = k-1$ and $|A| = k$ it follows that $|B \cap A| < |B| = k-1$; thus, $|A \setminus B| = |A|-|A \cap B| > k - (k-1) = 1$. Therefore, there are $e,f \in A \setminus B$ such that $e \neq f$. 
			It follows that there is $g \in \{e,f\}$ such that $\textnormal{\textsf{sum}}(B+g) = \textnormal{\textsf{sum}}(B)+g \neq \alpha$. We conclude from \eqref{eq:1} that $B+g \in \mathcal{K}$,
			implying that $B+g \in \cI(\Pi)$.
		\end{enumerate}
\end{enumerate}	\end{proof}
Observe that $\mathcal{K}_{n,k,\alpha} \cup \mathcal{L}_{n,k,\alpha}(\Pi)$ is the set of bases of the matroid. Moreover, for any arguments $n$, $k$, $\alpha$ and $\Pi$,
the cardinality of every  dependent set $S\in 2^{[n]}\setminus \cI_{n,k,\alpha}(\Pi)$  is at least the rank of the $\Pi$-matroid $M_{n,k,\alpha} (\Pi)=([n],\cI_{n,k,\alpha}(\Pi))$.
Such matroids are known as {\em paving matroids} (see, e.g.,~ \cite{Ox2006,PvdP14}).
Using $\Pi$-matroids, we define the following collection of oracle-EMB instances. In these instances, the matroid is a $\Pi$-matroid, 
where $\Pi$ is some unknown fixed family of subsets of the ground set. 
\begin{definition}
	\label{def:EMBpi}
	For every $n,k,\alpha \in \mathbb{N}_{>0}$, $\Pi \subseteq 2^{[n]}$ define the 
	{\em induced} oracle-\textnormal{EMB} instance of $n,k,\alpha,\Pi$, denoted $I_{n,k,\alpha}(\Pi)$, as follows. Let
	$\textnormal{\textsf{id}}_n: [n] \rightarrow [n]$, where $\textnormal{\textsf{id}}_n(i) = i~\forall i \in [n]$. Then, 
	$I_{n,k,\alpha}(\Pi) = ([n],\cI_{n,k,\alpha}(\Pi),\textnormal{\textsf{id}}_n,\alpha)$.
\end{definition}

Observe that the above is indeed an oracle-EMB instance if the independent sets of the given matroid are accessible via a membership oracle. The following is an easy consequence of \Cref{def:EMBpi}.

\begin{observation}
	\label{obs:YesNo}
	For every $n,k,\alpha \in \mathbb{N}_{>0}$ and $\Pi \subseteq 2^{[n]}$, it holds that $I_{n,k,\alpha}(\Pi)$ is an \textnormal{oracle-EMB} ''yes"-instance if and only if there is $S \in \Pi$ such that $|S| = k$ and $\textnormal{\textsf{sum}}(S) = \textnormal{\textsf{id}}_n(S) = \alpha$. 
\end{observation}

By \Cref{obs:YesNo}, an algorithm that finds an independent set of $\mathcal{M}_{n,k,\alpha}(\Pi)$ satisfying $|S| = k$ and $\textnormal{\textsf{sum}}(S) = \alpha$, 
in fact outputs a subset $S \in \Pi$. 
As the input for an induced oracle-EMB instance $I_{n,k,\alpha}(\Pi)$ does not contain an explicit encoding of $\Pi$, finding $S\in \Pi$ requires a sequence of queries to the membership oracle of $\mm$.  Roughly speaking, to decide $I_{n,k,\alpha}(\Pi)$, 
an algorithm for oracle-EMB must iterate over all (exponentially many) subsets $S \subseteq [n]$ such that $|S| = k$ and $\textnormal{\textsf{sum}}(S) = \alpha$. This is the intuition
behind the proof of the next result.

%% file: Pi-matroid.tex
\begin{figure}	
	\begin{center}
		\begin{tikzpicture}[ultra thick,scale=1.1, every node/.style={scale=1}]

		\node at (-1.5,3.5) {$\bf \textcolor{blue}{I}$};
		\draw (0,3) rectangle (6,4);
		\draw (1,3)--(1,4);
		\draw (2,3)--(2,4);
		\draw (3,3)--(3,4);
		
		\draw (5,3)--(5,4);
		
		\draw (6,3)--(6,4);								
		\node at (0.5,2.5) {$1$};
		\node at (0.5,3.5) {no};
		\node at (1.5,2.5) {$2$};
		\node at (1.5,3.5) {yes};
		\node at (2.5,2.5) {$3$};
		\node at (2.5,3.5) {no};
		\node at (4,3.5) {$\ldots$};
		\node at (5.5,2.5) {$|Q(b)|$};
		\node at (5.5,3.5) {no};
		
		\node at (8.5,3.5) {$\bf \textcolor{black}{S \notin Q(b)},~~ \textcolor{blue}{S \notin \cI(\Pi_{\emptyset})}$};
		
		
		\node at (-1.5,0.5) {$\bf \textcolor{red}{I_S}$};
		\draw (0,0) rectangle (6,1);
		\draw (1,0)--(1,1);
		\draw (2,0)--(2,1);
		\draw (3,0)--(3,1);
		
		\draw (5,0)--(5,1);
		
		
		\draw (6,0)--(6,1);								
		\node at (0.5,-0.5) {$1$};
		\node at (0.5,0.5) {no};
		\node at (1.5,-0.5) {$2$};
		\node at (1.5,0.5) {yes};
		\node at (2.5,-0.5) {$3$};
		\node at (2.5,0.5) {no};
		\node at (4,0.5) {$\ldots$};
		\node at (5.5,-0.5) {$|Q(b)|$};
		\node at (5.5,0.5) {no};
		
		\node at (8.5,0.5) {$\bf \textcolor{black}{S \notin Q(b)},~~ \textcolor{red}{ S \in \cI(\Pi_{S})}$};

		
		
		\end{tikzpicture}
	\end{center}
	\caption{\label{fig:proof} An illustration of the proof of \Cref{thm:emb}. The figure presents the sequences of queries to the membership oracles by the algorithm on the instances $I$ and $I_S$ for a string of bits $b$, such that $S \notin Q(b)$. The label ''yes" (''no") indicates that the queried set is (not) independent in the matroid. 
		The only query that distinguishes between $I$ and $I_S$ is on the set $S$, which is not queried; thus, the algorithm returns the same output for $I$ and $I_S$.}
\end{figure}
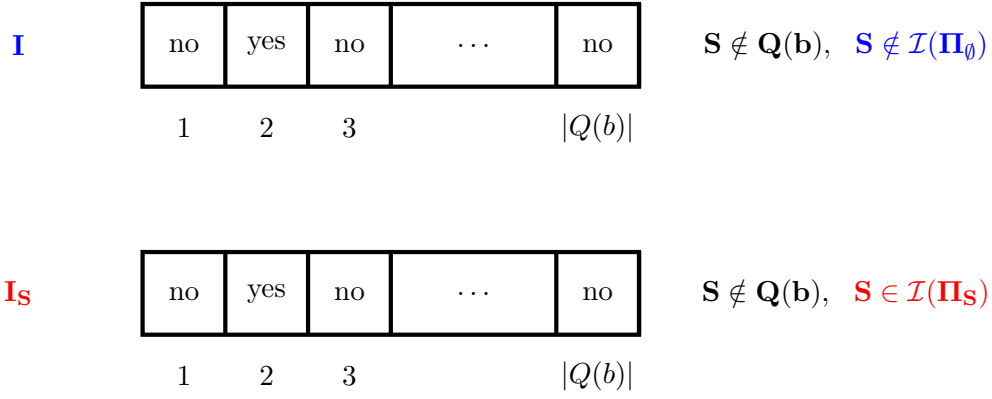

\thmemb*
	\begin{proof}
		Assume towards contradiction that there exist a constant $d \in \mathbb{N}$ and a randomized algorithm $\cA$ that decides every oracle-EMB instance $(E,\cI,c,T)$ in time $\left( (T+2) \cdot n \cdot m \right)^{d}$, where $n = |E|+1$ and $m = c(E)+1$.  
		For every $n,k,\alpha \in \mathbb{N}_{>0}$, consider the set of all subsets of $[n]$ with cardinality $k$ and sum $\alpha$: $$\mathcal{F}_{n,k,\alpha} = \big\{S \subseteq [n]~\big|~|S| = k, \textnormal{\textsf{sum}}(S) = \alpha\big\}.$$ 
		To reach a contradiction, we construct an induced oracle-EMB instance on which $\cA$ does not compute the proper output with sufficiently high probability. The parameters of the instance are extracted from the following combinatorial claim. 
			\begin{claim}
			\label{claim:H1}
			There are $\tilde{n} \in \mathbb{N}_{>0}$, $\tilde{k} \in \left[\tilde{n}\right]$, and $\tilde{\alpha} \in \left[\tilde{n}^2\right]$ such that 
$|\mathcal{F}_{\tilde{n},\tilde{k},\tilde{\alpha}}|> 	2 \cdot \left(	12 \cdot \tilde{n}^5 \right)^{d} $. 
		\end{claim}
		\begin{claimproof}
				Since $d$ is a constant, 
				there is $\tilde{n} \in \mathbb{N}_{>0}$ such that 
				\begin{equation}
					\label{eq:nTild}
						\left(	12 \cdot \tilde{n}^5 \right)^{d} < \frac{2^{\tilde{n}}-1}{2 \cdot \tilde{n}^3}.
				\end{equation}
				  Fix $\tilde{n} \in \mathbb{N}_{>0}$ satisfying \eqref{eq:nTild}. Recall that $\sum_{k \in \{0,1,\ldots,\tilde{n}\}} {\tilde{n} \choose k} = 2^{\tilde{n}}$ from a basic property of the Pascal triangle; therefore, $\sum_{k \in \{1,\ldots,\tilde{n}\}} {\tilde{n} \choose k} = 2^{\tilde{n}}-1$. Thus, there is  $\tilde{k} \in \{1,\ldots,\tilde{n}\}$, 
				  such that 
				  \begin{equation}
				  	\label{eq:kTild}
				  	{\tilde{n} \choose \tilde{k}} \geq \frac{2^{\tilde{n}}-1}{\tilde{n}}.
				  \end{equation}
				   Fix $\tilde{k} \in \left[\tilde{n}\right]$ satisfying \eqref{eq:kTild}. Observe that for each $S \in 2^{[\tilde{n}]}$ satisfying $|S| = \tilde{k}>0$, it holds that $1 \leq \textnormal{\textsf{sum}}(S) \leq |S| \cdot \max_{i \in [\tilde{n}]} i = \tilde{n}^2$. Thus, there are $\tilde{n}^2$ possibilities for $\alpha \in \left[\tilde{n}^2\right]$ satisfying $\alpha = \textnormal{\textsf{sum}}(S)$, for some $S \in 2^{[\tilde{n}]}$ such that $|S| = \tilde{k}$. Moreover, there are ${\tilde{n} \choose \tilde{k}}$ subsets of $[\tilde{n}]$ of cardinality $\tilde{k}$. By the pigeonhole principle, there is $\tilde{\alpha} \in  \left[  \tilde{n}^2 \right]$ such that $|\mathcal{F}_{\tilde{n},\tilde{k},\tilde{\alpha}}| \geq  \frac{{\tilde{n} \choose \tilde{k}}}{\tilde{n}^2}$. Thus,
				   $$|\mathcal{F}_{\tilde{n},\tilde{k},\tilde{\alpha}}| \geq \frac{{\tilde{n} \choose \tilde{k}}}{\tilde{n}^2}  \geq  \frac{2^{\tilde{n}}-1}{\tilde{n} \cdot \tilde{n}^2} > 2 \cdot \left(	12 \cdot \tilde{n}^5 \right)^{d}.$$ The second inequality follows from \eqref{eq:kTild}, and the third inequality holds by \eqref{eq:nTild}. 
		\end{claimproof}

		Let $\tilde{n} \in \mathbb{N}_{>0}$, $\tilde{k} \in \left[\tilde{n}\right]$, and $\tilde{\alpha} \in \left[\tilde{n}^2\right]$  satisfying the conditions of \Cref{claim:H1}. Define $t$ to be the maximum running time of $\cA$ on an induced EMB instance $I_{\tilde{n},\tilde{k},\tilde{\alpha}}(\Pi)$ over all $\Pi \in 2^{\left[ \tilde{n}\right]}$. 
			\begin{claim}
			\label{claim:H2}
			$t \leq	\left(	12 \cdot \tilde{n}^5 \right)^{d}$.    
		\end{claim}
		\begin{claimproof}
				Let $T = \tilde{\alpha}$, $n = \tilde{n}+1$, and $m = \textsf{id}_{\tilde{n}}(\left[\tilde{n}\right])+1$. By the running time guarantee of $\cA$, it follows that $t \leq	\left(	n \cdot (T+2) \cdot m \right)^{d}$. It remains to bound $n \cdot (T+2) \cdot m$. Since $\textsf{id}_{\tilde{n}}(\left[\tilde{n}\right]) = \textsf{sum}(\left[\tilde{n}\right])$ and $\tilde{\alpha} \leq \tilde{n}^2$, 
				 \begin{equation*}
				\label{eq:InducedParameters}
				n  \cdot (T+2) \cdot m \leq	 \left(\tilde{n}+1\right) \cdot \left( \tilde{n}^2+2\right)\left( \textsf{sum}( \left[\tilde{n} \right])+1 \right) \leq 2 \tilde{n} \cdot 3 \tilde{n}^2 \cdot \left(\frac{\tilde{n} \cdot (\tilde{n}+1)}{2}+1 \right) \leq 12 \cdot \tilde{n}^5. 
			\end{equation*}  The second inequality follows from the sum of the terms of an arithmetic sequence. By the above and the running time guarantee of $\cA$, it follows that $t \leq	\left(	n \cdot (T+2) \cdot m \right)^{d} \leq 	\left(	12 \cdot \tilde{n}^5 \right)^{d}$.   
		\end{claimproof}
		
	Given  an induced oracle-EMB instance $I_{\tilde{n},\tilde{k},\tilde{\alpha}}(\Pi)$, for some $\Pi \subseteq \left[ \tilde{n}\right]$,  the randomized algorithm $\cA$ produces a random string of bits $\bar{b} \in \{0,1\}^{t}$ and performs a sequence of queries to the membership oracle of $M_{\tilde{n},\tilde{k},\tilde{\alpha}}$, based on $\bar{b},\tilde{n},\tilde{k},\tilde{\alpha}$, and the results of the previous queries; then, the algorithm decides the given instance based on the queries.

		 Let $\Pi_{\emptyset} = \emptyset$, and consider the induced oracle-EMB instance $I = I_{\tilde{n},\tilde{k},\tilde{\alpha}}(\Pi_{\emptyset})$. Given a string of bits $b \in \{0,1\}^{t}$, let $Q(b) \subseteq 2^{[\tilde{n}]}$ be the set of  all subsets $S \subseteq [\tilde{n}]$ queried by $\cA$ on the instance $I$ and the bit-string $b$ (on the membership oracle of $M_{\tilde{n},\tilde{k},\tilde{\alpha}}(\Pi_{\emptyset})$). For clarity, we use $\bar{b}$ for a random string, and $b$ for a realization of $\bar{b}$ to a specific string. Note that $Q(b)$ is a set, since the algorithm is deterministic for every $b \in \{0,1\}^{t}$; conversely, $Q(\bar{b})$ is a random set for a random string $\bar{b} \in \{0,1\}^{t}$. As the running time of $\cA$ on 
		$I$ is bounded by $t$, it holds that $|Q(b)| \leq t$ for every $b \in \{0,1\}^{t}$. Let 
		\begin{equation*}
			\label{eq:R}
			R(I) = \left\{S \in \mathcal{F}_{\tilde{n},\tilde{k},\tilde{\alpha}}~\bigg|~\Pr \left(S \in Q(\bar{b}) \right) \geq \frac{1}{2} \right\}
		\end{equation*} be all sets in $ \mathcal{F}_{\tilde{n},\tilde{k},\tilde{\alpha}}$ that are queried by $\cA$ with probability  at least $\frac{1}{2}$. 
			\begin{claim}
			\label{claim:H3}
			$|R(I)|  < |\mathcal{F}_{\tilde{n},\tilde{k},\tilde{\alpha}}|$. 
		\end{claim}
		\begin{claimproof}
			By the definition of $R(I)$, it holds that 
				\begin{equation*}
				\label{eq:R0}
				\begin{aligned}
					|R(I)| ={} & \left| \left\{S \in \mathcal{F}_{\tilde{n},\tilde{k},\tilde{\alpha}}~\bigg|~\Pr \left(S \in Q(\bar{b}) \right) \geq \frac{1}{2} \right\} \right| \\
					\leq{} & 2 \cdot \sum_{S \in \mathcal{F}_{\tilde{n},\tilde{k},\tilde{\alpha}}} \Pr \left(S \in Q(\bar{b}) \right)  \\
					={} & 	 2 \cdot \sum_{S \in \mathcal{F}_{\tilde{n},\tilde{k},\tilde{\alpha}}~} \sum_{b \in \{0,1\}^{t}} \Pr \left(S \in Q(b) \right) \cdot \Pr(\bar{b} = b)
				\end{aligned}
			\end{equation*} Thus, by changing the order of summation, we have 
				\begin{equation*}
				\begin{aligned}
					|R(I)| \leq 2 \cdot  \sum_{b \in \{0,1\}^{t}~} \Pr(\bar{b} = b) \cdot \sum_{S \in \mathcal{F}_{\tilde{n},\tilde{k},\tilde{\alpha}}} \Pr \left(S \in Q(b) \right) 
					\leq 2 \cdot  \sum_{b \in \{0,1\}^{t}~} \Pr(\bar{b} = b) \cdot |Q(b)| .
				\end{aligned}
			\end{equation*} 
			
Since $\left|Q(b)\right| \leq t$ for all $b \in \{0,1\}^t$, by the above we have
			\begin{equation*}
				|R(I)| 	\leq 2  t \cdot \sum_{b \in \{0,1\}^{t}~} \Pr (\bar{b} = b) = 2  \cdot t \leq 2\cdot	\left(	12 \cdot \tilde{n}^5 \right)^{d} 
				< |\mathcal{F}_{\tilde{n},\tilde{k},\tilde{\alpha}}|. 
			\end{equation*}
			The second inequality follows from \Cref{claim:H2}. The last inequality holds by \Cref{claim:H1}. 
		\end{claimproof}
	
	By \Cref{claim:H3}, there exists $S \in F_{\tilde{n},\tilde{k},\tilde{\alpha}} \setminus R(I)$. 
		Consider the induced oracle-EMB instance 
		$I_S = I_{\tilde{n},\tilde{k},\tilde{\alpha}}(\Pi_S)$ where $\Pi_{S} = \{S\}$, and let $B = \{b \in \{0,1\}^t~|~ S \notin Q(b)\}$ be all strings for which $S$ is not queried by $\cA$ on the instance $I$ .
	Observe that for all $b \in B$ it holds that the answers to all queries for $T \in Q(b)$ are the same for both oracles (of $M_{\tilde{n},\tilde{k},\tilde{\alpha}}(\Pi_{\emptyset})$ and $M_{\tilde{n},\tilde{k},\tilde{\alpha}}(\Pi_{S})$). Moreover, the decision on which set to query next depends only on $b,\tilde{n},\tilde{k},\tilde{\alpha}$, and the answers to previous queries.  
	
	Hence, for all $b \in B$,  the executions of $\cA$ on the instances $I$ and $I_S$ are identical.
		Since $\Pi_{\emptyset} = \emptyset$, by 
		\Cref{obs:YesNo}, $I$ is a ``no''-instance for oracle-EMB; thus, $\cA$ returns that $I_S$ is a ``no''-instance for every $b\in B$. However, since $\textsf{id}_n(S) = \textsf{sum}(S) = \alpha$, $|S| = k$, and $S \in \Pi_{S}$, it follows that $I_S$ is a ``yes''-instance by \Cref{obs:YesNo}.  
		%
		%
		Therefore, $\cA$ does not decide $I_S$ correctly for all $b \in B$.  We give an illustration in \Cref{fig:proof}. Since  $S \notin R(I)$, it holds that $\Pr \left( \bar{b} \in B \right) = \Pr \left(S \notin Q(\bar{b}) \right) > \frac{1}{2}$; thus, with probability greater than $\frac{1}{2}$, $\cA$ does not decide correctly the instance $I_S$. 
		%
		A contradiction to the correctness of  $\cA$  as a randomized algorithm for oracle-EMB. The statement of the theorem follows. 
	\end{proof}

%% file: HardnessOracle.tex
\section{Hardness of Matroid Optimization with a Linear Constraint}
\label{sec:tight_absolute}
In this section we use \Cref{thm:emb} 
to prove \Cref{thm:main}. We apply the following reductions from EMB to MOL problems. Recall that $\cQ = \cP \setminus \{(\min,\IS,\leq)\}$ is the set of parameters for non-trivial MOL problems. Given a $P$-MOL problem for some $P \in \cQ$, and an EMB instance $I$, the reduction returns an instance $R_P(I)$ of the $P$-MOL problem. Note that the reduction is purely mathematical, and 
does not specify the encoding of the instance. This will be useful for obtaining our hardness results in non-oracle computational models (see \Cref{sec:SAT}). 

Given $(\opt,\cF,\triangleleft), (\opt',\cF',\triangleleft') \in \cQ$, we use the notation $(\opt \textnormal{ is } \opt')$,  $(\cF \textnormal{ is } \cF')$, $(\triangleleft \textnormal{ is } \triangleleft')$ to denote the boolean expressions of equality between parameters of a MOL problem. For example, $(\opt \textnormal{ is } \opt')$ is {\em true} if and only if either $\opt,\opt'$ are both $\max$, or $\opt,\opt'$ are both $\min$. 


\begin{definition}
	\label{def:reduction}
	Given an \textnormal{EMB} instance $I =(E,\cI,c,T)$ and  $P \in \cQ$ where $P = (\opt,\cF,\triangleleft)$, define the
	 {\em reduced $P$-MOL instance} of $I$, denoted by  $R_P(I) = (E,\cI,v_I,w_{I,P},L_{I,P})$, as follows.
	\begin{enumerate}
		\item Define the auxiliary variable 
		\begin{align*}
			d(P) = & \begin{cases}
				0 & ~~\textnormal{if  } \bigg((\opt \textnormal{ is } \max)  \textnormal{ and }   (\triangleleft \textnormal{ is}   \leq) \bigg) \textnormal{ or } 
				\bigg((\opt \textnormal{ is } \min)  \textnormal{ and }   (\triangleleft \textnormal{ is}   \geq) \bigg)\\
				1 & ~~\textnormal{otherwise.}
			\end{cases}
		\end{align*} For example, if $P = (\max,\IS,\leq)$ then $d(P) = 0$, and if  $P' = (\max,\IS,\geq)$ then $d(P') = 1$. 
		\item Let $H_I = 2 \cdot \max \{1, c(E)\}$.
		\item For all $e \in E$ let $v_I(e) =  H_I +c(e)$. 
		\item For all $e \in E$ let $w_{I,P}(e) =  H_I +c(e) \cdot (-1)^{d(P)}$.
		\item Let $k_I = \max_{S \in \cI} |S|$ be the rank of $(E,\cI)$. 
		\item Define $L_{I,P} = k_I \cdot H_I +T \cdot  (-1)^{d(P)}$. 
	\end{enumerate}
\end{definition}

Now, for every EMB instance $I =(E,\cI,c,T)$, define the {\em error parameter} of $I$ as \begin{equation}
\label{eq:epsI}
\eps_I =  \frac{1}{8 \cdot (|E|+1) \cdot (T+1) \cdot \left(c(E)+1\right)}.
\end{equation}
Indeed, since the value selected for of the error parameter is sufficiently small, we can use a $(1+\eps_I)$-approximation for $R_P(I)$ to decide an EMB instance $I$.  
\begin{restatable}{theorem}{thmH}
	\label{thm:H}
	Given an instance $I =(E,\cI,c,T)$ of \textnormal{EMB}, and $P \in \cQ$ with $P = (\opt,\cF,\triangleleft)$, the following holds. 
	\begin{enumerate}
		\item If there is a solution $S$ for $R_P(I)$ such that $v_I(S) = k_I \cdot H_I +T$, then $I$ is a ``yes"-instance.  \label{h:1}
		\item  If $I$ is a ``yes"-instance then: 
		$(i)$ $R_P(I)$ has a solution, and $(ii)$ every $(1+\eps_I)$-approximate solution $S$ for $R_P(I)$ satisfies $v_I(S) = k_I \cdot H_I +T$. \label{h:2}
	\end{enumerate}
\end{restatable}
Using \Cref{thm:H} and an assumed randomized Fully PTAS for the $P$-MOL problem, we can decide EMB in time which contradicts \Cref{thm:emb}. This gives the proof of \Cref{thm:main}. We first prove \Cref{thm:main}, and later give the proof of \Cref{thm:H}.   

 \thmMain*

\begin{proof}
	Assume towards a contradiction that there is a randomized Fully PTAS $\cA$ for oracle $P$-MOL. We use $\cA$ to decide oracle-EMB. Let $I = (E,\cI,c,T)$  be an oracle-EMB instance, and 
	consider the following randomized algorithm $\cB$ that decides $I$. 
	\begin{enumerate}
		\item Construct the oracle $P$-MOL instance $R_{P}(I)$
		with the membership oracle of $(E,\cI)$.\label{step:1} 
		\item 
		Execute $\cA$ with the input $R_{P}(I)$ and $\eps_I$. \label{step:2} 
		\item If $\cA$ returns that $R_{P}(I)$ does not have a solution $-$ Return ''no" on $I$. \label{step:3}  
		\item Otherwise, let $S \leftarrow \cA(R_{P}(I),\eps_I)$ be the solution returned  by $\cA$.\label{step:4} 
		\item Return ''yes" on $I$ if and only if $v_I(S) = H_I \cdot k_I+T$. \label{step:5} 
	\end{enumerate}
	
	Let $n = |E|+1$ and $m = c(E)+1$. Note that $R_{P}(I)$ can be naively constructed from $I$ in time $\left(n \cdot (T+2) \cdot m\right)^{O(1)}$ using \Cref{def:reduction}. 
	As $\cA$ is a randomized Fully PTAS for oracle $P$-MOL, and by the selection of the error parameter \eqref{eq:epsI}, the running time of $\cB$ on $I$ is $\left(n \cdot (T+2) \cdot m\right)^{O(1)}$. 
	We now show correctness. 
	\begin{itemize}
		\item If $\cB$ returns ``yes" on $I$, then Step~\ref{step:4} of the algorithm computes a solution $S$ for $R_p(I)$ satisfying $v_I(S) = H_I \cdot k_I+T$. Thus, by \Cref{thm:H}, $I$ is a ``yes" instance. 

		\item If $I$ is a ``yes" instance then $R_P(I)$ has a solution by \Cref{thm:H}.
		As $\cA$ is a randomized Fully PTAS, with probability at least $\frac{1}{2}$ $\cA$ returns a $(1+\eps_I)$-approximate solution $S$ for $R_P(I)$ in Step~\ref{step:4}. By \Cref{thm:H}, $v_I(S) = H_I \cdot k_I+T$ (with probability at least $\frac{1}{2}$). Thus, $\cB$ returns ``yes" on $I$ with probability at least $\frac{1}{2}$.
		\end{itemize}  Hence, $\cB$ is a randomized algorithm which decides the oracle-EMB instance $I$ in time \newline $\left(n \cdot (T+2) \cdot m\right)^{O(1)}$. This is a contradiction to \Cref{thm:emb}.  
\end{proof}

In the remainder of this section we prove \Cref{thm:H}. We start with some basic properties of the reduction outlined in \Cref{def:reduction}.   

\begin{lemma}
\label{lem:FD2}
Given an instance $I =(E,\cI,c,T)$ of \textnormal{EMB},
let $P \in \cQ$ and consider a solution $S$ for $R_P(I)$ satisfying $v_I(S) = k_I \cdot H_I+T$. Then, $S$ is a solution for $I$.  
\end{lemma}

\begin{proof}
	Let $\cm = (E,\cI)$. As $S$ is a solution for $R_P(I)$, it holds that $S \in \IS(\cm)$; thus, $|S| \leq k_I$. Assume towards contradiction that $|S|<k_I$. Then, 
	$$v_I(S) = |S| \cdot H_I+c(S) \leq (k_I-1) \cdot H_I+c(S) \leq (k_I-1) \cdot H_I+c(E)<k_I \cdot H_I \leq k_I \cdot H_I+T.$$ 
	We reach a contradiction since $v_I(S) = k_I \cdot H_I+T$; thus, $|S| = k_I$, and  
	\begin{equation}
		\label{eq:S=k}
		k_I \cdot H_I+T = v_I(S) =   |S| \cdot H_I+c(S) =  k_I \cdot H_I+c(S). 
	\end{equation}  As $|S| = k_I$, we have that $S$ is a basis of $\cm$, and by \eqref{eq:S=k}, $c(S) = T$. Hence, $S$ is a solution for $I$.   
\end{proof}

The next result is the converse of the statement in \Cref{lem:FD2}.   

	\begin{lemma}
	\label{lem:FD1}
	Let $S$ be a solution for a given \textnormal{EMB} instance $I =(E,\cI,c,T)$,
	and let $P \in \cQ$ where $P = (\opt,\cF,\triangleleft)$. Then, $S$ is a solution for $R_P(I)$ of value $v_I(S) = k_I \cdot H_I+T$.  
\end{lemma}

\begin{proof}
	Let $\cm = (E,\cI)$. Since $S$ is a solution for $I$ we have $S \in \bases(\cm)$; thus, $S \in \cF(\cm)$. Then,    
	$$w_{I,P}(S) = |S| \cdot H_I+c(S) \cdot (-1)^{d(P)} =  k_I \cdot H_I+T \cdot (-1)^{d(P)} = L_{I,P}.$$
	The second equality holds since $S$ is a solution for $I$; thus, $|S| = k_I$ (as $S$ is a basis of $\cm$), and $c(S) = T$. We conclude that $S$ is a solution for $R_P(I)$. Finally, note that $S$ satisfies 
\begin{align*}
~~~~~~~~~~~~~~~~~~~~~~~~~~~~~	v_I(S) = |S| \cdot H_I + c(S)
	= k_I \cdot H_I +T ~~~~~~~~~~~~~~~~~~~~~~~~~~~~~~~~~~~~~~~~~~~~~~~~~~~~~~~~ \qedhere
\end{align*}
\end{proof}

The next claim gives an upper bound on the optimal value for maximization MOL problems. We then derive an analogous lower bound for minimization (non-trivial) MOL problems.

	\begin{lemma}
	\label{lem:OptMaxBound}
	Let $I =(E,\cI,c,T)$ be an \textnormal{EMB} instance and $P \in \cQ$, where $P = (\opt,\cF,\triangleleft)$ and $\left(\opt \textnormal{ is } \max\right)$. Then, for every solution $S$ of $R_P(I)$ it holds that $v_I(S) \leq k_I \cdot H_I+T$.  
\end{lemma}

\begin{proof}
	Let $S$ be an optimal solution for $R_P(I)$. Thus, $|S| \leq k_I$ as $S \in \cI$. If $|S| < k_I$ then
	 \begin{equation*}
			v_I(S) \leq (k_I-1) \cdot H_I+c(S) \leq (k_I-1) \cdot H_I+c(E) <k_I \cdot H_I \leq k_I \cdot H_I+T.
		\end{equation*}
	Otherwise, $|S| = k_I$. Consider the two cases for $\triangleleft$.
			\begin{enumerate}
		\item $(\triangleleft \textnormal{ is}   \leq)$. Then, $d(P) = 0$ (see \Cref{def:reduction}); thus, since $S$ is a solution for $R_P(I)$: $$v_I(S) = w_{I,P}(S) \leq L_{I,P} = k_I \cdot H_I+T.$$ 
		\item $(\triangleleft \textnormal{ is}   \geq)$. Then, $d(P) = 1$. As $S$ is a solution for $R_P(I)$,
		\begin{equation}
			\label{eq:Dp=1}
			  k_I \cdot H_I-c(S)=|S|\cdot H_I-c(S) =w_{I,P}(S) \geq L_{I,P} = k_I \cdot H_I-T.
		\end{equation} 
		By \eqref{eq:Dp=1}, it follows that $c(S) \leq T$; thus, $	v_I(S) = k_I \cdot H_I+c(S) \leq k_I \cdot H_I+T$. 
	\end{enumerate} 
	
In all the above cases, we have that $v_I(S) \leq k_I \cdot H_I+T$, implying the statement of the lemma. 
\end{proof}

Now, for minimization problems we have the next result. 
	\begin{lemma}
	\label{lem:OptMinBound}
	Let $I =(E,\cI,c,T)$ be an \textnormal{EMB} instance, and $P \in \cQ$, where $P = (\opt,\cF,\triangleleft)$ and $\left(\opt \textnormal{ is } \min\right)$. Then, for every solution $S$ of $R_P(I)$, it holds that $v_I(S) \geq k_I \cdot H_I+T$.  
\end{lemma}

\begin{proof}
	Let $S$ be a solution for $R_P(I)$. We consider two cases. 
	\begin{enumerate}
		\item $(\triangleleft \textnormal{ is}   \geq)$. Then, $d(P) = 0$; since $S$ is a solution for $R_P(I)$, $$v_I(S) = w_{I,P}(S) \geq L_{I,P} = k_I \cdot H_I+T.$$ 
		\item  $(\triangleleft \textnormal{ is}   \leq)$. Then, $d(P) = 1$; since $P$ is a non-trivial MOL problem (i.e., $P \in \cQ$), it follows that $\left(\cF \textnormal{ is } \bases\right)$. Thus, $S$ is a solution for $R_P(I)$ such that $|S| = k_I$, and
		\begin{equation}
			\label{eq:Dp=2}
			 	  k_I \cdot H_I-c(S)= |S| \cdot H_I-c(S) =w_{I,P}(S) \leq L_{I,P} = k_I \cdot H_I-T.
		\end{equation} 
		By \eqref{eq:Dp=2} it follows that $c(S) \geq T$; thus, $v_I(S) = k_I \cdot H_I+c(S) \geq k_I \cdot H_I+T$. 
	\end{enumerate}
In all of the above cases, we have that $v_I(S) \geq k_I \cdot H_I+T$, implying the statement of the lemma. 
\end{proof}

Using \Cref{lem:FD1,lem:FD2,lem:OptMaxBound,lem:OptMinBound}, we can now prove \Cref{thm:H}. 
\thmH*
\begin{proof}
We note that Property~\ref{h:1} follows directly from \Cref{lem:FD2}. For Property~\ref{h:2}, assume that $I$ is a ``yes"-instance, then by \Cref{lem:FD1}, there is a solution $D$ for $R_P(I)$ such that $v_I(D) = H_I \cdot k_I+T$. It remains to show Property~\ref{h:2}. $(ii)$.
Let $S$ be a $(1+\eps_I)$-approximate solution for $R_P(I)$. We distinguish between two cases. 
	
		\begin{enumerate}
		\item $(\opt \textnormal{ is}   \max)$. Then, by \Cref{lem:OptMaxBound},
		\begin{equation*}
			\label{eq:Max1}
			0 \leq H_I \cdot k_I+T-v_I(S).
		\end{equation*} Moreover, since $S$ is a $(1+\eps_I)$-approximate solution for $R_P(I)$, and $D$ is a solution for $R_P(I)$,
		\begin{equation*}
			\label{eq:MaxApprox}
	H_I \cdot k_I+T-v_I(S) \leq H_I \cdot k_I+T - \frac{v_I(D)}{(1+\eps_I)} =  \frac{\eps_I \cdot \left(H_I \cdot k_I+T \right) }{(1+\eps_I)} \leq \eps_I \cdot \left(H_I \cdot k_I+T \right). 
		\end{equation*} 
	By the above, it follows that
	\begin{equation*}
			  \left|	v_I(S)-\left(H_I \cdot k_I+T\right) \right|	\leq 	\eps_I \cdot \left(H_I \cdot k_I+T\right). 
	\end{equation*}
		\item  $(\opt \textnormal{ is}   \min)$. This case is analogous to the above.  By \Cref{lem:OptMinBound}, 
		\begin{equation*}
			\label{eq:Min1}
	0 \leq v_I(S) - \left(H_I \cdot k_I+T\right). 
		\end{equation*} Since $S$ is a $(1+\eps_I)$-approximate solution for $R_P(I)$, and $D$ is a solution for $R_P(I)$,
		\begin{equation*}
			\label{eq:MinApprox}
		v_I(S) - \left(H_I \cdot k_I+T\right) \leq (1+\eps_I) \cdot v_I(D) - \left(H_I \cdot k_I+T\right) =   \eps_I \cdot \left(H_I \cdot k_I+T \right). 
		\end{equation*} 
		By the above, 
		\begin{equation*}
			\left|	v_I(S)-\left(H_I \cdot k_I+T\right) \right|	\leq 	\eps_I \cdot \left(H_I \cdot k_I+T\right). 
		\end{equation*}

	\end{enumerate} 
	Thus in both cases it holds that, 
	\begin{equation}
		\label{eq:combined}
		\left|	v_I(S)-\left(H_I \cdot k_I+T\right) \right|	\leq 	\eps_I \cdot \left(H_I \cdot k_I+T\right). 
	\end{equation}
	Let $n = |E|+1$ and $m = c(E)+1$. Then, by the selection of $\eps_I$ in \eqref{eq:epsI}, 
	\begin{equation}
		\label{eq:MinMax}
		\eps_I \cdot \left(H_I \cdot k_I+T\right) = \frac{H_I \cdot k_I+T}{8 \cdot n \cdot (T+1) \cdot m} 
		\leq\frac{2 \cdot m \cdot n + T}{8 \cdot n \cdot (T+1) \cdot m} 
		<\frac{4 \cdot m \cdot n \cdot (T+1)}{8 \cdot n \cdot (T+1) \cdot m} 
		= \frac{1}{2}. 
	\end{equation}
The first inequality holds since $k_I \leq |E|$ and $H_I \leq 2 \cdot m$.	Therefore, by \eqref{eq:combined} and \eqref{eq:MinMax}, 
	\begin{equation}
		\label{eq:FinalMinMax}
		  \left|	v_I(S)-\left(H_I \cdot k_I+T\right) \right|	\leq 	\eps_I \cdot \left(H_I \cdot k_I+T\right) <  \frac{1}{2}.
	\end{equation} Since $v_I(S) \in \mathbb{N}$ by \Cref{def:reduction}, it follows from \eqref{eq:FinalMinMax} that $v_I(S) = H_I \cdot k_I+T$. This gives
	the statement of the theorem. 
\end{proof}

%% file: SAT-matroid.tex
\section{Lower Bounds in the Standard Computational Model}
\label{sec:SAT}
Our hardness result in \Cref{sec:Pi} shows that {\sc oracle Exact Matroid Basis (EMB)} is hard, leading to the unconditional lower bounds for all non-trivial oracle MOL problems in~\Cref{sec:tight_absolute}. Nonetheless, these hardness results consider matroids with general membership oracles, and do not give a lower bound for matroids that can be efficiently encoded. This is particularly important, as in some settings oracle models differ from non-oracle models w.r.t complexity~\cite{canetti2004random,chang1994random}. Moreover, some matroids show up in problems that can be encoded efficiently.
 This includes partition matroids, graphic matroids, linear matroids, etc (see, e.g., \cite{Sc03} for a survey on various families of matroids). Next, we formally define an efficient encoding of matroids.

	\begin{definition}
		\label{def:encode}
		
		A function $f:\{0,1\}^* \rightarrow 2^{\mathbb{N}} \times 2^{2^{\mathbb{N}}}$ is called {\em matroid decoder} if for every $I \in \{0,1\}^*$ it holds that $f(I) = \left(E_{f(I)}, \cI_{f(I)} \right)$ is a matroid, and the following holds. \begin{enumerate}
			\item There is an algorithm that given $I \in \{0,1\}^*$ returns $E_{f(I)}$ in time $|I|^{O(1)}$.\label{cond:1}
			\item There is an algorithm that given $I \in \{0,1\}^*$ and $S \subseteq E_{f(I)}$ decides if $S \in \cI_{f(I)}$ in time $|I|^{O(1)}$. \label{cond:2}
		\end{enumerate} 
	\end{definition}

	There is a simple matroid decoder that can decode every matroid $(E,\cI)$ (such that $E\subseteq \mathbb{N}$), in which the encoding $I$ explicitly lists $\cI$. 
	However, using such a matroid decoder, the encoding size of a matroid might be very large, up to $|I|=\Omega\left(2^{|E|}\right)$, while we often seek algorithms with running times polynomial in $|E|$. 
	One way to overcome this difficulty is via the oracle model considered  in previous sections. However, our results in this model may suggest that the hardness of EMB and MOL problems is 	due to the intrinsic hardness of the oracle model.
	 Yet, there are families of matroids with very efficient encoding. For example, a {\em uniform matroid} $(E,\cI)$, where $\cI = \{S \subseteq E~|~|S| \leq k\}$, can be efficiently encoded using $I = (E,k) \in \{0,1\}^*$. Clearly, the time to decide membership of a given subset $S \subseteq E$ depends only on $|S|$ and $k$. Another example is given in \Cref{fig:M1}.

We start with a definition of an encoded variant of EMB. We technically define a different problem for every decoder $f$. The definition of the problem is analogous to the versions of EMB considered earlier in the paper, besides that the matroid is given via an arbitrary bit-string $I \in \{0,1\}^*$, that a matroid decoder $f$ decodes into a matroid $f(I)$. 

\problem{{\sc $f$-decoded Exact Matroid Basis} ($f$-decoded EMB)}{
	{\bf Decoder} & $f:\{0,1\}^* \rightarrow 2^\mathbb{N} \times 2^{2^{\mathbb{N}}}$ is a matroid decoder.\\
	{\bf Instance} &  $(I,c,T)$, where $I \in \{0,1\}^*$, $c:E_{f(I)} \rightarrow \mathbb{N}$, $T \in \mathbb{N}$.\\
	{\bf Solution} & A basis $S$ of the matroid $f(I)$ such that $c(S)=T$.\\
	{\bf Objective} & Decide if there is a solution.
}

As a simple example, consider the $f_{\textsf{u}}$-decoded EMB problem, for a specific matroid decoder $f_{\textsf{u}}$ that decodes uniform matroids. The matroid decoder $f_{\textsf{u}}$ interprets every $I \in \{0,1\}^*$ as $I = (E,k)$ where $E$ is a set (of numbers) and $k \in \mathbb{N}$, and returns the uniform matroid $f_{\textsf{u}}(I) = (E,\cI)$ such that $\cI = \{S \subseteq E~|~|S| \leq k\}$; clearly, $f_{\textsf{u}}$ is a matroid decoder. Thus, an instance of $f_{\textsf{u}}$-decoded EMB is a tuple $U = ((E,k),c,T)$ and a solution of $U$ is $S \subseteq E$ such that $|S| = k$ and $c(S) = T$; the goal, as before, is to decide if there is a solution. This problem is commonly known as the $k$-{\sc subset sum}.

	\begin{figure}
	\centering
	
	\begin{tikzpicture}[thick]
		\tikzstyle{edge} = [->, line width=1pt]
		\node (v1) at (0,2) {};
		\node (v5) at (9,5) {};
		\node (v6) at (9,4.5) {};
		\node (v7) at (9,4) {};

		\node (v2) at (1.5,3) {};
		\node at (2,1) {$\bf \textcolor{black}{U}$};
		\node (v3) at (4,2.5) {};
		\node (v4) at (7,2) {};
		\node (v5) at (7,3) {};
		\node at (7,1) {$\bf \textcolor{black}{V}$};
		\node at (12,1) {$\bf \textcolor{black}{W}$};
		\node (u0) at (11,2.5) {};
		\node (u1) at (12,2.5) {};
		\node (u2) at (13,2.5) {};

		\begin{scope}[fill opacity=0.9]
			\filldraw[fill=yellow!70] ($(v1)+(-0.5,0)$) 
			to[out=90,in=180] ($(v2) + (0,0.5)$) 
			to[out=0,in=90] ($(v3) + (1,0)$)
			to[out=270,in=0] ($(v2) + (1,-0.8)$)
			to[out=180,in=270] ($(v1)+(-0.5,0)$);

			\filldraw[fill=green!50] ($(v4)+(-1.5,0.2)$)
			to[out=90,in=180] ($(v4)+(1,2)$)
			to[out=0,in=90] ($(v4)+(0.6,0.3)$)
			to[out=270,in=0] ($(v4)+(1,-0.6)$)
			to[out=180,in=270] ($(v4)+(-1.5,0.2)$);

				\filldraw[fill=red!70] ($(u1)+(-3,0)$) 
				to[out=90,in=180] ($(u2) + (0,0.5)$) 
				to[out=0,in=90] ($(u0) + (4,0)$)
				to[out=270,in=0] ($(u2) + (0,-0.8)$)
				to[out=180,in=270] ($(u1)+(-3,0)$);
			\end{scope}

			\filldraw[fill=brown!70] (v1) circle (0.1) node [right] {$u_1$};
			\filldraw[fill=darkgray!70] (v2) circle (0.1) node [below left] {$u_2$};
			\filldraw[fill=blue!70] (v3) circle (0.1) node [left] {$u_3$};
			\filldraw[fill=red!70] (v4) circle (0.1) node [below] {$v_1$};
			\filldraw[fill=purple!70] (v5) circle (0.1) node [below] {$v_2$};
			\filldraw[fill=white!70] (u1) circle (0.1) node [right] {$w_3$};
			\filldraw[fill=gray!70] (u2) circle (0.1) node [right] {$w_2$};
			\filldraw[fill=pink!70] (u0) circle (0.1) node [left] {$w_1$};
		\end{tikzpicture}
		\caption{\label{fig:M1} An example of a {\em partition matroid} $(E,\cI)$, which  can be efficiently encoded. The ground set is $E = \{u_1,u_2,u_3,v_1,v_2,w_1,w_2,w_3\}$, partitioned into three sets: $U,V,W$. The independent sets are all subsets of $E$ containing at most one element from $U,V$, and $W$; that is, $\cI = \{S \subseteq E~|~\forall X \in \{U,V,W\}: |S \cap X| \leq 1\}$. A simple efficient encoding of $(E,\cI)$ is $I = (E,U,V,W)$. 
			Membership can be decided efficiently given $I$, by checking the feasibility
			of a given set $S$ w.r.t. $U,V$ and $W$.} 
	\end{figure}
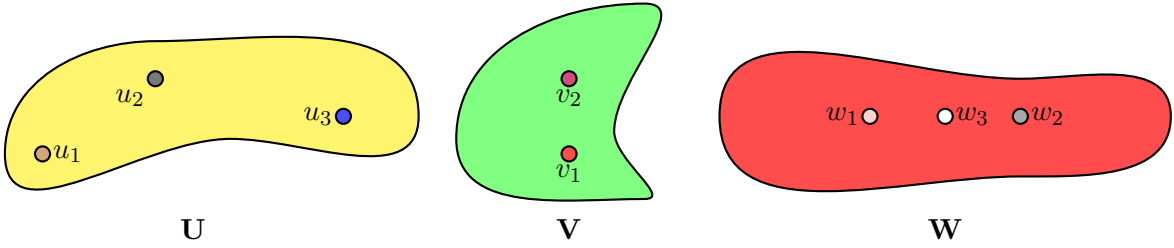

Recall that \Cref{thm:emb} asserts that oracle-EMB does not admit a pseudo-polynomial time algorithm. However, 
 this does not rule out that hypothetically, for every matroid decoder $f$ there is a pseudo-polynomial time algorithm  
 for $f$-decoded EMB. The next result excludes this option. 

\begin{theorem}
	\label{thm:embA}
		Assuming $\textnormal{P} \neq \textnormal{NP}$, there is a matroid decoder $f$ such that there is no algorithm  for $f$-\textnormal{decoded EMB} that for any  $f$-\textnormal{decoded EMB} instance $U =(I,c,T)$, where $n = \left|E_{f(I)}\right|+1$ and $m = c\left(E_{f(I)}\right) +1$, runs in time $ (n \cdot(T+2) \cdot m)^{O(1)}$. 
\end{theorem}

The proof of \Cref{thm:embA} is given towards the end of this section. Analogously to our hardness result for oracle MOL problems, we use \Cref{thm:embA} to give a hardness result for an encoded version of MOL problems. For every matroid decoder $f$ and every $P \in \cP$, we define a variant of the $P$-{\sc matroid optimization with a linear constraint ($P$-MOL)} problem 
	in which the matroid is given via an arbitrary bit-string which a matroid decoder $f$ decodes into a matroid. 
	Formally, let $P \in \cP$, where $P = (\opt,\cF,\triangleleft)$, be the parameters of the $P$-MOL problem. For a matroid decoder $f$, we define the $f$-decoded $P$-MOL problem as follows.  
	
		\problem{ {\sc $f$-decoded $P$-matroid optimization with a linear constraint} (f-decoded $P$-MOL)}{
		{\bf Decoder} & $f:\{0,1\}^* \rightarrow 2^\mathbb{N} \times 2^{2^{\mathbb{N}}}$ is a matroid decoder.\\
		{\bf Instance} &   $(I,v,w,L)$, where $I \in \{0,1\}^*$,  $v:E_{f(I)} \rightarrow \mathbb{R}_{\geq 0}$, $w:E_{f(I)} \rightarrow\mathbb{R}_{\geq 0}$, $L \in \mathbb{R}_{\geq 0}$. \\
		{\bf Objective} & $\opt~ v(S) \text{ s.t. } S \in \cF\left(f(I)\right),~  w(S) \triangleleft L$.
	}
	
	For example, consider the encoded version of the 
	$P$-MOL for $P = (\max,\IS,\leq)$ with the matroid decoder $f_{\textsf{u}}$ that decodes uniform matroids. 
	An instance of the $f_{\textsf{u}}$-decoded $P$-MOL problem is a tuple $U = (I,v,w,L)$ where $I = (E,k)$ is a bit-string used for extracting the uniform matroid $f_{\textsf{u}}(I) = (E,\cI)$ such that $\cI = \{S \subseteq E~|~|S| \leq k\}$, $v$ is the value function, $w$ is the weight function, and $L$ is the bound. 
	A {\em solution} of $U$ is $S \subseteq E$ such that $|S| \leq k$ and $w(S) \leq L$; the goal is to find a solution $S$ of maximum value $v(S)$. This problem is widely known as {\sc knapsack with cardinality constraint}.  
	
	Recall that $\cQ = \cP \setminus \{(\min,\IS,\leq)\}$ is the set of parameters for non-trivial MOL problems. Using the hardness of $f$-decoded, for some matroid decoder $f$ (details on $f$ are given towards the end of the section), we show the hardness of the $f$-decoded  variant of all non-trivial MOL problems.

	 \begin{theorem}
	 	\label{thm:mainA}
	 	Assuming $\textnormal{P} \neq \textnormal{NP}$, for any $P \in \cQ$ there is a matroid decoder $f$ such that there is no \textnormal{Fully PTAS} for $f$-\textnormal{decoded} $P$-\textnormal{MOL}. 
	 \end{theorem}
	 
	 In the remainder of this section, we prove \Cref{thm:embA} and  \Cref{thm:mainA}. The matroid decoder used in our proofs decodes a subclass of the $\Pi$-matroid family (see \Cref{sec:Pi}), in which the secret family $\Pi$ consists of the solutions for a  {\sc boolean satisfiability problem (SAT)} instance.

	In a SAT instance $A = (V,\bar{V},\cC)$ with $n \in \mathbb{N}$ variables (in a slightly simplified notation), we are given a set $V = \{v_1, \ldots,v_n\}$ of variables, their negations $\bar{V} = \{\bar{v}_1,\ldots, \bar{v}_n\}$, and a set $\cC \subseteq 2^{V \cup \bar{V}}$ of clauses. 
	The goal is to decide if there is a set $S \subseteq [n]$ 
	satisfying that for all $C \in \cC$ there is $i \in [n]$ such that one of the following holds.   \begin{itemize}
		\item $v_i \in C$ and $i \in S$.
		\item $\bar{v}_i \in C$ and $i \notin S$. 
	\end{itemize} Such a set $S$ is called a {\em solution} of $A$; let $\mathcal{S}(A)$ be the set of solutions of a SAT instance $A$. In addition, let $n(A) = n$ be the number of variables in the instance $A$. 
	The family of SAT-matroids is the subfamily of $\Pi$-matroids where $\Pi = \mathcal{S}(A)$ for some SAT instance $A$ (for the notation and definition of $\Pi$-matroids, see \Cref{def:Matroid}). Specifically, 
	
	\begin{definition}
		\label{def:SATmatroid}
		Let $A$ be a \textnormal{SAT} instance, $k \in \left[n(A)\right]$, and $\alpha \in \left[  {n(A)}^{2}\right]$. Define the {\bf SAT-matroid} on $A,k,\alpha$ as $M_{n(A),k,\alpha}(\mathcal{S}(A)) = \left(\left[n(A)\right],\cI_{n(A),k,\alpha}(\mathcal{S}(A))\right)$. 
	\end{definition}
	We show below that SAT-matroids can be  encoded efficiently. 
	For every $I \in \{0,1\}^*$, we interpret $I$ as $I = (A,k,\alpha)$, where $A$ is a \textnormal{SAT} instance, $k \in \left[n(A) \right]$, and $\alpha \in \left[ {n(A)}^2 \right]$; w.l.o.g., we may assume that every $I = (A,k,\alpha) \in \{0,1\}^*$ can be interpreted in that manner; moreover, we can assume that $n(A) \leq |A|$, where $|A|$ is the encoding size of $A$. The following definition defines a matroid decoder that decodes SAT-matroids from a bit-string. 
	\begin{definition}
		\label{def:SATencoder}
		Define the {\em SAT-decoder}  as the function $f_{\textnormal{SAT}}:\{0,1\}^* \rightarrow 2^\mathbb{N} \times 2^{2^{\mathbb{N}}}$ such that for all $I = (A,k,\alpha) \in \{0,1\}^*$ it holds that $f_{\textnormal{SAT}}(I) = M_{n(A),k,\alpha}(\mathcal{S}(A))$. 
	\end{definition}
	In the next result, we show that the SAT-decoder $\fs$ is indeed a  matroid decoder.
	\begin{lemma}
		\label{lem:encodable}
		$\fs$ is a matroid decoder such that $|I| = |A|^{O(1)}$ for every $I=(A,k,\alpha) \in \{0,1\}^*$. 
	\end{lemma}
	Recall the sets  $\mathcal{J}_{n,k}$,  $\mathcal{K}_{n,k,\alpha}$, $\mathcal{L}_{n,k,\alpha}(\Pi)$ and $\cI_{n,k,\alpha}(\Pi)$ were defined in  \Cref{def:Matroid}.
	\begin{proof} 
		We first show that $\fs$ is a matroid decoder. 
		%
		%
		For all $I = (A,k,\alpha) \in \{0,1\}^*$  it holds that $\mathcal{S}(A) \subseteq 2^{\left[n(A)\right]}$; thus, by \Cref{lem:1} it holds that $f_{\textnormal{SAT}}(I)$ returns a matroid. Moreover, for all $I = (A,k,\alpha) \in \{0,1\}^*$ it holds that $k,\alpha = {n(A)}^{O(1)}$ and $n(A) \leq |A|$ by our encoding; hence, $|I| = |A|^{O(1)}$.  It remains to prove that $\fs$ satisfies \Cref{cond:1,cond:2} in \Cref{def:encode}.
		\begin{enumerate}
			\item  	\Cref{cond:1}. Given $I = (A,k,\alpha) \in \{0,1\}^*$, we can compute $E_{f_{\textnormal{SAT}}(I)} = \left[n(A)\right]$ in time $|I|^{O(1)}$ by identifying the number of variables (i.e., $n(A)$) in $A$. 
			\item 		\Cref{cond:2}. Let $I = (A,k,\alpha) \in \{0,1\}^*$ and let $n = n(A)$; we show that we can decide membership in the matroid $\fs(I)$ in time $\io$.  
			For every $S \subseteq \left[n\right]$ we can compute $|S|$ and $\textsf{sum}(S)$ in time $|I|^{O(1)}$. Thus, validating if $S \in \mathcal{J}_{n,k}$ or $S \in \mathcal{K}_{n,k,\alpha}$ can be computed in time $\io$ regardless of $\mathcal{S}(A)$. In addition, 
			deciding if $S \in \mathcal{S}(A)$ can be computed in time $|A|^{O(1)} = |I|^{O(1)}$ by confirming that all clauses are satisfied. Thus, by \eqref{eq:1} we can decide if $S \in \mathcal{L}_{n,k,\alpha}(\mathcal{S}(A))$. Overall, deciding whether $S \in \mathcal{J}_{n,k}$, $S \in \mathcal{K}_{n,k,\alpha}$, or $S \in \mathcal{L}_{n,k,\alpha}(\mathcal{S}(A))$ can be computed in time $|I|^{O(1)}$. By the above and \eqref{eq:1}, we can compute in time $\io$ correctly if a set $S$ belongs to $\cI_{\fs (I)} = \cI_{n,k,\alpha}(\mathcal{S}(A))$. 
		\end{enumerate}
		Thus, $\fs$ is a matroid decoder by \Cref{def:encode}, which gives the statement of the lemma.  
	\end{proof}

	We will use an algorithm for $f$-decoded EMB to obtain an algorithm for SAT. To this end, consider the following family of {\em structured } $\fs$-decoded EMB instances.

											\begin{definition}
			\label{def:problemE}
			An $\fs$-{\sc decoded EMB} instance $(I,c,T)$, where $I = (A,k,\alpha)$ is called {\em structured} if  $\alpha \in \left[ {n(A)}^2 \right]$, $T = \alpha$, and $c: \left[ n(A) \right] \rightarrow \mathbb{N}$ such that for all $i \in \left[ n(A) \right]$ it holds that $c(i) = i$ 
		\end{definition}

		The next observation immediately follows from the definition of structured $\fs$-decoded EMB instances and SAT-matroids. 
		
		\begin{observation}
			\label{obs:YesNoD}
			for any structured $\fs$-{\sc decoded EMB} instance $U = (I,c,T)$, where $I = (A,k,\alpha)$, and $S \subseteq \left[ n(A)\right]$, it holds that: $S$ is a solution for $U$ if and only if $S \in \mathcal{S}(A)$, $|S| = k$, and $\textnormal{\textsf{sum}}(S) = c(S) = \alpha$. 
		\end{observation}
		
		We show that given a polynomial algorithm that decides structured $\fs$-decoded EMB instances, we can decide SAT. This result easily imply \Cref{thm:embA} as shown afterwards.  
		

		\begin{lemma}
				\label{lem:SAT}
				Assuming $\textnormal{P} \neq \textnormal{NP}$, there no algorithm that decides every $\fs$-decoded \textnormal{EMB} structured instance $(I,c,T)$ in time $\io$. 
			\end{lemma}
	\begin{proof}
	
	Assume that $\textnormal{P} \neq \textnormal{NP}$ and assume towards a contradiction that there is an algorithm $\cA$ that decides every $\fs$-decoded EMB structured instance $(I,c,T)$ in time $\io$. We describe the following algorithm $\cB$ which decides SAT using $\cA$. Let $A$ be a SAT instance and let $n = n(A)$. 
	\begin{enumerate}
		\item 	For every $k \in \left[n\right]$ and $\alpha \in \left[n^2\right]$, let $I_{k,\alpha} = (A,k,\alpha)$; execute $\cA$ on the structured instance $U_{k,\alpha} = (I_{k,\alpha},c,\alpha)$. 
		\item  If $\emptyset$ is a solution for $A$,  return ``yes" on $A$. 
		\item If there are $k \in \left[n\right]$ and $\alpha \in \left[n^2\right]$ such that $\cA$ returns ``yes" on $U_{k,\alpha}$, 
	 return ``yes" on $A$. 
	 \item Otherwise, return ``no" on $A$. 
	\end{enumerate}
	To show that $\cB$ decides SAT, consider the following cases. 
	\begin{itemize}
		\item $A$ is a ``yes"-instance. Then, there is a solution $S \in \mathcal{S}(A)$. If $\emptyset \in \mathcal{S}(A)$ then $\cB$ returns ``yes" on $A$. Otherwise, there are $k \in [n]$, $\alpha \in \left[n^2\right]$, and a solution $S \in \mathcal{S}(A)$ such that $|S| = k$ and $ \textnormal{\textsf{sum}}(S) = \alpha$. Since $U_{k,\alpha}$ is structured, we have $c(S) = \alpha$; therefore, $S$ is a solution for $U_{k,\alpha}$ by \Cref{obs:YesNoD}. Consequently, $\cA$ returns ``yes" on $U_{k,\alpha}$; hence, $\cB$ returns ``yes" on $A$.

		\item $A$ is a ``no"-instance. Then, $\mathcal{S}(A) = \emptyset$. 
		Assume towards a contradiction that $\cB$ returns ``yes" on $A$. Thus, since $\emptyset \notin \mathcal{S}(A)$, there are $k \in [n]$ and $\alpha \in \left[ n^2\right]$ such that $\cA$ returns ``yes" on $U_{k,\alpha}$, implying that there is a solution $S$ for $U_{k,\alpha}$. Since $U_{k,\alpha}$ is structured, by \Cref{obs:YesNoD} it holds that $S \in \mathcal{S}(A)$. 
		This is a contradiction that $\mathcal{S}(A) = \emptyset$. Thus, $\cB$ returns ``no" on $A$. 
	\end{itemize} 
	Hence, $\cB$ decides the SAT instance $A$ correctly. We analyze the running time below.  
	
	 By our encoding, it holds that $n \leq |A|$ as $A$ contains $n$ distinct variables. In addition, observe that the identity function $c:[n] \rightarrow [n]$ can be encoded using at most $n \cdot \log (n)$ bits (this is the weight function of every structured instance $U_{k,\alpha}$ constructed in the algorithm, for $k \in [n]$ and $\alpha \in \left[n^2\right]$). Then, for all $k \in [n]$ and $\alpha \in \left[ n^2\right]$ observe that the encoding of the instance $U_{k,\alpha}$ satisfies: $$|U_{k,\alpha}| = |A|+k+\alpha+n \cdot \log (n) \leq |A|+n+n^2 +n \cdot \log (n)= O\left(|A|^2\right).$$ 
	 Moreover, the number of iterations (i.e., the number of distinct options for the parameters $k,\alpha$) is bounded by $n \cdot n^2 = n^3 = O(|A|^3)$. Thus, the running time of each iteration is bounded by $|U_{k,\alpha}|^{O(1)} = |A|^{O(1)}$ by the running time guarantee of $\cA$. Therefore, our algorithm decides SAT in time $|A|^{O(1)}$. Since SAT is known to be NP-Hard (see, e.g., \cite{garey1997computers} for more details), we reach a contradiction to the existence of $\cA$. 
\end{proof}

From the above result, the hardness of the more general $\fs$-decoded EMB easily follows. As an immediate corollary, \Cref{thm:embD} gives the proof of \Cref{thm:embA}. 

	\begin{restatable}{theorem}{thmembD}
	\label{thm:embD}
	Assuming $\textnormal{P} \neq \textnormal{NP}$, there is no algorithm  for $\fs$-decoded \textnormal{EMB} that runs in time $ (n \cdot(T+2) \cdot m)^{O(1)}$, for any  $\fs$-decoded \textnormal{EMB} instance $U =(I,c,T)$ where $n = \left|E_{\fs(I)}\right|+1$ and $m = c\left(E_{\fs(I)}\right) +1$. 
\end{restatable}

\begin{proof}
	Let $U = (I,c,T)$ be a structured $\fs$-decoded EMB instance and let $I = (A,k,\alpha)$. Recall that by our encoding it holds that $|I| \geq |A| \geq n(A)$. Thus, 
	since $U$ is structured, it follows that  $\left|E_{\fs(I)}\right| = n(A) \leq |I|$ and 
	$$c\left(E_{\fs(I)}\right) = \textsf{sum}\left(\left[n(A)\right]\right) \leq \frac{n(A) \cdot (n(A)+1)}{2} = O\left(|I|^2\right).$$
	Thus, by the above, an algorithm $\cA$ that decides every $\fs$-EMB  instance $(I',c',T')$ in time $(n' \cdot(T'+2) \cdot m')^{O(1)}$, where  $n' = |E_{\fs(I')}|+1$ and $m' = c\left(E_{\fs(I)}\right) +1$, in particular decides $U$ in time $\io$. As $U$ is an arbitrary structured instance, unless $\textnormal{P} = \textnormal{NP}$, by \Cref{lem:SAT} such an algorithm $\cA$ cannot exist.  
\end{proof}


	Finally, 
	we show that 
the variants of non-trivial matroid optimization with a linear constraint (MOL) problems, in which the decoding is performed by the SAT-decoder $\fs$, do not admit Fully PTAS under the standard assumption $\textnormal{P} \neq \textnormal{NP}$.   The proof is similar to the proof of \Cref{thm:main} in \Cref{sec:tight_absolute}. \Cref{thm:mainD} directly gives the proof of \Cref{thm:mainA}.

\begin{restatable}{lemma}{thmMainD}
	\label{thm:mainD}
	Assuming $\textnormal{P} \neq \textnormal{NP}$, for any $P \in \cQ$ there is no \textnormal{Fully PTAS} for $\fs$-\textnormal{decoded} $P$-\textnormal{MOL}. 
\end{restatable}

\begin{proof}
	Assume that $\textnormal{P} \neq \textnormal{NP}$ and assume towards a contradiction that there is a Fully PTAS $\cA$ for $\fs$-decoded $P$-MOL. We use $\cA$ to decide $\fs$-decoded EMB. Let $U = (I,c,T)$  be an $\fs$-decoded EMB instance. 
	Consider the following algorithm $\cB$ that decides $U$. In the algorithm, we use the reduction described in \Cref{def:reduction} and an error parameter defined in \eqref{eq:epsI}. 
	\begin{enumerate}
		\item Construct the $\fs$-decoded $P$-MOL instance $X = (I,v_J,w_{J,P},L_{J,P})$, where \newline $J = \left(E_{\fs(I)},\cI_{\fs (I)},c,T\right)$ is an EMB instance.\label{step:1D} 
		\item 
		Execute $\cA$ with the input $X$ and $\eps_J$ (see \eqref{eq:epsI}). \label{step:2D} 
		\item If $\cA$ returns ``no" on $X$: Return ``no" on $U$. 
		\item Otherwise, let $S \leftarrow \cA(X,\eps_J)$ be the returned solution by $\cA$.\label{step:4D} 
		\item Return ``yes" on $U$ 
		if and only if $v_J(S) = H_J \cdot k_J+T$. \label{step:5D} 
	\end{enumerate}
	
	Observe that $U,J$ share the same matroid, weight function, and target value; the only difference between them is the representation of the matroid, that is encoded in $U$ while in $J$ the representation of the matroid is not specified. Similarly, $X$ and $R_P(J)$ also share the same matroid, value function, weight function, and bound; the only difference between them is that $X$ gives a concrete encoding of the matroid.  
	
	Let $n = |E|+1$ and let $m = c(E)+1$. Note that $X$ can be trivially constructed from $U$ in time $\left(n \cdot (T+2) \cdot m\right)^{O(1)}$ by \Cref{def:reduction}. 
	Thus, the running time of $\cB$ on $U$ is $\left(n \cdot (T+2) \cdot m\right)^{O(1)}$, since $\cA$ is a Fully PTAS for $\fs$-decoded $P$-MOL and by the selection of the error parameter  $\eps_J$ in \eqref{eq:epsI}. 
	We now show correctness. 

	\begin{itemize}
		\item If $\cB$ returns ``yes" on $U$. 
		Then, Step 4 of the algorithm computes a solution $S$ for $X$ satisfying $v_J(S) = H_J \cdot k_J+T$. Clearly, $S$ is a solution for $R_P(J)$ as well. By \Cref{thm:H}, 
		$S$ is also a solution for $J$. As $J$ and $U$ share the same matroid, weight function, and target value, we conclude that $S$ is a solution for $U$. Therefore, $U$ is a ``yes"-instance.

		\item If $U$ is a ``yes"-instance. Then, $J$ has a solution. Therefore, by \Cref{thm:H} it holds that $R_P(J)$ has a solution, which implies that $X$ has a solution. Thus, since $\cA$ is a Fully PTAS for $\fs$-decoded $P$-MOL, $\cA$ returns a $(1+\eps_J)$-approximate solution $S$ for $X$; clearly, $S$ is also a $(1+\eps_J)$-approximate solution for $R_P(J)$. 
		Then, by \Cref{thm:H} it follows that $v_J(S) = H_J \cdot k_J+T$. Thus, $\cB$ returns ``yes" on $U$.
	\end{itemize}  Hence, $\cB$ decides $\fs$-decoded EMB in time $\left(n \cdot (T+2) \cdot m\right)^{O(1)}$, contradicting \Cref{thm:embD}.  
\end{proof}

%% file: discussion.tex
\section{Discussion}
\label{sec:discussion}

In this paper, we derive lower bounds for the family of matroid optimization problems with a linear constraint. We show that none of the (non-trivial) members of this family 
admits a Fully PTAS. In particular, this rules out a Fully PTAS for previously studied problems such as {\sc budgeted matroid independent set}, 
{\sc constrained minimum basis of a matroid},
and {\sc knapsack cover with a matroid}. 
As BM and CMB admit an Efficient PTAS,
our lower bounds resolve the complexity status of these problems, which has been
open also for the generalization of {\sc budgeted matroid intersection} \cite{CVZ11,BBGS11,DKS23b}. Our preliminary study shows that using the techniques of \cite{HL04}, we may be able to derive Efficient PTAS for {\em all}  MOL problems. This would imply that \Cref{thm:main} gives a tight lower bound for the entire MOL family.
We leave the details for future work.  

 A key result of this paper is that {\sc exact matroid basis} (EMB) does not admit a pseudo-polynomial time algorithm, unlike the known special cases of $k$-{\sc subset sum} and EMB on a linear matroid.
  Our proofs can be used to obtain lower bounds for other problems. For example, 
  the hardness result for EMB can be adapted to yield  
  lower bounds for related {\em parameterized} problems~\cite{FGKSS23,doron2023budgeted}.
  Moreover, the proof of \Cref{thm:main} can be modified to show that an  Efficient PTAS for a non-trivial MOL problem with running time $f\left( \frac{1}{\eps}\right)\cdot \textnormal{poly}(n)$ must satisfy $f\left(\frac{1}{\eps} \right)=\Omega\left( 2^{\eps^{-\frac{1}{4}}}\right)$. We leave these generalizations of our results to a later version of this paper.   

 Our results build on the $\Pi$-matroid family introduced in this paper. Such matroids exploit the interaction between a weight function and the underlying matroid constraint of the given problem. Aside from the implications of our results for previously studied problems,
 the new subclass of $\Pi$-matroids may enable to derive lower bounds for 
 other problems. For example, consider the generalization of BM where the objective function is submodular and monotone.
 This is known as monotone submodular maximization with a knapsack and a matroid constraint \cite{chekuri2010dependent}. Indeed, if the knapsack constraint is removed, there is a tight $\left(1-\frac{1}{e} \right)$-approximation for the problem \cite{calinescu2011maximizing}. The same bound holds if we relax the matroid constraint~\cite{sviridenko2004note}. However, the best known approximation for the problem with a knapsack and a matroid constraint
 is $\left(1-\frac{1}{e} -\eps \right)$~\cite{chekuri2010dependent}. This setting resembles the status of MOL problems prior to our work, where removing either the linear or the matroid constraint induces a substantially easier problem. 
 The potential use of $\Pi$-matroid variants to rule out a
 $\left(1-\frac{1}{e} \right)$-approximation for 
 the above problem remains an interesting open question. One additional consequence of our lower bound in a different domain is that solving {\em configuration} LPs for packing problems with a matroid constraint (e.g., \cite{grage2019eptas,doron2023afptas}) cannot be solved for general matroids in FPTAS time using the standard ellipsoid methods.

We show unconditional hardness results in the oracle model (even if randomization is allowed), and give analogous lower bounds where the matroids are encoded as part of the input, assuming $\textnormal{P} \neq \textnormal{NP}$. Our construction in \Cref{sec:SAT} can be used to derive hardness results for other matroid problems in non-oracle models.
Specifically, we can obtain in the standard computational model
hardness results analogous to those in the oracle model of \cite{JK82}.
This includes a proof that it is NP-hard to decide if a given matroid is uniform, analogous to the unconditional hardness result in the oracle model of \cite{JK82}. We leave these results for future work. 

Our lower bounds for MOL problems on general matroids call for a more comprehensive study of these problems on restricted classes of matroids. We note the existence of Fully PTASs for MOL problems on some restricted matroid classes, e.g., BM on a laminar matroid
 or KCM on a partition matroid.
The question whether (non-trivial) MOL problems admit Fully PTAS on broader matroid classes,
 such as graphical matroids or linear matroids, remains open.
In particular, it would be interesting to obtain a Fully PTAS for {\sc constrained minimum spanning tree} \cite{HL04} and BM on a linear matroid $-$ or show that one does not exist.

				\newpage

%% file: main.bbl
\begin{thebibliography}{10}
\providecommand{\url}[1]{\texttt{#1}}
\providecommand{\urlprefix}{URL }
\providecommand{\doi}[1]{https://doi.org/#1}

\bibitem{andersen1996bicriterion}
Andersen, K.A., J{\"o}rnsten, K., Lind, M.: On bicriterion minimal spanning
  trees: An approximation. Computers \& Operations Research  \textbf{23}(12),
  1171--1182 (1996)

\bibitem{BBGS11}
Berger, A., Bonifaci, V., Grandoni, F., Sch{\"a}fer, G.: Budgeted matching and
  budgeted matroid intersection via the gasoline puzzle. Mathematical
  Programming  \textbf{128}(1),  355--372 (2011)

\bibitem{bringmann2021fine}
Bringmann, K., Nakos, V.: A fine-grained perspective on approximating subset
  sum and partition. In: Proceedings of the 2021 ACM-SIAM Symposium on Discrete
  Algorithms (SODA). pp. 1797--1815. SIAM (2021)

\bibitem{calinescu2011maximizing}
Calinescu, G., Chekuri, C., Pal, M., Vondr{\'a}k, J.: Maximizing a monotone
  submodular function subject to a matroid constraint. SIAM Journal on
  Computing  \textbf{40}(6),  1740--1766 (2011)

\bibitem{CGM92}
Camerini, P.M., Galbiati, G., Maffioli, F.: Random pseudo-polynomial algorithms
  for exact matroid problems. Journal of Algorithms  \textbf{13}(2),  258--273
  (1992)

\bibitem{canetti2004random}
Canetti, R., Goldreich, O., Halevi, S.: The random oracle methodology,
  revisited. Journal of the ACM (JACM)  \textbf{51}(4),  557--594 (2004)

\bibitem{caprara2000approximation}
Caprara, A., Kellerer, H., Pferschy, U., Pisinger, D.: Approximation algorithms
  for knapsack problems with cardinality constraints. European Journal of
  Operational Research  \textbf{123}(2),  333--345 (2000)

\bibitem{CCNR13}
Chakaravarthy, V.T., Choudhury, A.R., Natarajan, S.R., Roy, S.: Knapsack cover
  subject to a matroid constraint. In: Proc. FSTTCS (2013)

\bibitem{chang1994random}
Chang, R., Chor, B., Goldreich, O., Hartmanis, J., H{\aa}stad, J., Ranjan, D.,
  Rohatgi, P.: The random oracle hypothesis is false. Journal of Computer and
  System Sciences  \textbf{49}(1),  24--39 (1994)

\bibitem{chekuri2005polynomial}
Chekuri, C., Khanna, S.: A polynomial time approximation scheme for the
  multiple knapsack problem. SIAM Journal on Computing  \textbf{35}(3),
  713--728 (2005)

\bibitem{chekuri2010dependent}
Chekuri, C., Vondr{\'a}k, J., Zenklusen, R.: Dependent randomized rounding via
  exchange properties of combinatorial structures. In: 2010 IEEE 51st Annual
  Symposium on Foundations of Computer Science. pp. 575--584. IEEE (2010)

\bibitem{CVZ11}
Chekuri, C., Vondr{\'a}k, J., Zenklusen, R.: Multi-budgeted matchings and
  matroid intersection via dependent rounding. In: Proc. SODA. pp. 1080--1097
  (2011)

\bibitem{CLRS22}
Cormen, T.H., Leiserson, C.E., Rivest, R.L., Stein, C.: Introduction to
  algorithms. MIT press (2022)

\bibitem{cygan2019problems}
Cygan, M., Mucha, M., Wegrzycki, K., WIodarczyk, M.: On problems equivalent to
  (min,+)-convolution. ACM Transactions on Algorithms (TALG)  \textbf{15}(1),
  1--25 (2019)

\bibitem{DKS23c}
Doron-Arad, I., Kulik, A., Shachnai, H.: An {FPTAS} for budgeted laminar
  matroid independent set. Operations Research Letters  \textbf{51}

\bibitem{doron2023afptas}
Doron-Arad, I., Kulik, A., Shachnai, H.: An afptas for bin packing with
  partition matroid via a new method for lp rounding. In: Approximation,
  Randomization, and Combinatorial Optimization. Algorithms and Techniques
  (APPROX/RANDOM 2023). Schloss-Dagstuhl-Leibniz Zentrum f{\"u}r Informatik
  (2023)

\bibitem{doron2023budgeted}
Doron-Arad, I., Kulik, A., Shachnai, H.: Budgeted matroid maximization: a
  parameterized viewpoint. IPEC  (2023)

\bibitem{DKS23b}
Doron-Arad, I., Kulik, A., Shachnai, H.: An {EPTAS} for budgeted matching and
  budgeted matroid intersection via representative sets. In: Proc. {ICALP}. pp.
  49:1--49:16 (2023)

\bibitem{DKS23}
Doron-Arad, I., Kulik, A., Shachnai, H.: An {EPTAS} for budgeted matroid
  independent set. In: Proc. SOSA. pp. 69--83 (2023)

\bibitem{edmonds1965minimum}
Edmonds, J.: Minimum partition of a matroid into independent subsets. J. Res.
  Nat. Bur. Standards Sect. B  \textbf{69},  67--72 (1965)

\bibitem{FGKSS23}
Fomin, F.V., Golovach, P.A., Korhonen, T., Simonov, K., Stamoulis, G.:
  Fixed-parameter tractability of maximum colored path and beyond. In: Proc.
  SODA. pp. 3700--3712 (2023)

\bibitem{garey1997computers}
Garey, M.R.: Computers and intractability: A guide to the theory of
  np-completeness, freeman. Fundamental  (1997)

\bibitem{grage2019eptas}
Grage, K., Jansen, K., Klein, K.M.: An eptas for machine scheduling with
  bag-constraints. In: The 31st ACM Symposium on Parallelism in Algorithms and
  Architectures. pp. 135--144 (2019)

\bibitem{GZ10}
Grandoni, F., Zenklusen, R.: Approximation schemes for multi-budgeted
  independence systems. In: Proc. ESA. pp. 536--548 (2010)

\bibitem{HL04}
Hassin, R., Levin, A.: An efficient polynomial time approximation scheme for
  the constrained minimum spanning tree problem using matroid intersection.
  SIAM Journal on Computing  \textbf{33}(2),  261--268 (2004)

\bibitem{hong2004fully}
Hong, S.P., Chung, S.J., Park, B.H.: A fully polynomial bicriteria
  approximation scheme for the constrained spanning tree problem. Operations
  Research Letters  \textbf{32}(3),  233--239 (2004)

\bibitem{JK82}
Jensen, P.M., Korte, B.: Complexity of matroid property algorithms. SIAM
  Journal on Computing  \textbf{11}(1),  184--190 (1982)

\bibitem{karp1985complexity}
Karp, R.M., Upfal, E., Wigderson, A.: The complexity of parallel computation on
  matroids. In: 26th Annual Symposium on Foundations of Computer Science (sfcs
  1985). pp. 541--550. IEEE (1985)

\bibitem{kulik2010there}
Kulik, A., Shachnai, H.: There is no eptas for two-dimensional knapsack.
  Information Processing Letters  \textbf{110}(16),  707--710 (2010)

\bibitem{La75}
Lawler, E.L.: Matroid intersection algorithms. Mathematical programming
  \textbf{9}(1),  31--56 (1975)

\bibitem{La79}
Lawler, E.L.: Fast approximation algorithms for knapsack problems. Math. Oper.
  Res.  \textbf{4}(4),  339--356 (1979)

\bibitem{lovasz1978matroid}
Lov{\'a}sz, L.: The matroid matching problem. Algebraic methods in graph theory
   \textbf{2},  495--517 (1978)

\bibitem{Ox2006}
Oxley, J.G.: Matroid theory, vol.~3. Oxford University Press, USA (2006)

\bibitem{PvdP14}
Pendavingh, R., Van~der Pol, J.: On the number of matroids compared to the
  number of sparse paving matroids. arXiv preprint arXiv:1411.0935  (2014)

\bibitem{RG96}
Ravi, R., Goemans, M.X.: The constrained minimum spanning tree problem. In:
  Algorithm Theory—SWAT'96: 5th Scandinavian Workshop on Algorithm Theory
  Reykjav{\'\i}k, Iceland, July 3--5, 1996 Proceedings 5. pp. 66--75. Springer
  (1996)

\bibitem{Sc03}
Schrijver, A.: Combinatorial optimization: polyhedra and efficiency, vol.~24.
  Springer (2003)

\bibitem{Si79}
Sinha, P., Zoltners, A.A.: The multiple-choice knapsack problem. Operations
  Research  \textbf{27}(3),  503--515 (1979)

\bibitem{soto2014simple}
Soto, J.A.: A simple ptas for weighted matroid matching on strongly base
  orderable matroids. Discrete Applied Mathematics  \textbf{164},  406--412
  (2014)

\bibitem{sviridenko2004note}
Sviridenko, M.: A note on maximizing a submodular set function subject to a
  knapsack constraint. Operations Research Letters  \textbf{32}(1),  41--43
  (2004)

\end{thebibliography}
